\documentclass[hidelinks]{amsart}

\usepackage{xcolor,enumerate}
\usepackage[utf8]{inputenc}
\usepackage[T1]{fontenc}
\usepackage{hyperref}
\usepackage{url}
\usepackage{booktabs}
\usepackage{nicefrac}
\usepackage{microtype}
\usepackage{amsmath,amssymb,amsfonts,latexsym}
\usepackage{bm}
\usepackage{graphicx}
\usepackage{textcomp}
\usepackage{subcaption}
\usepackage{color}
\usepackage{mathrsfs}
\usepackage{pgfplots}
\usepackage{mathtools}
\usepackage{comment}
\usepackage{algorithm}
\usepackage{algorithmic}

\usepackage{tikz}
\definecolor{revisiongreen}{RGB}{0,110,70} 
\definecolor{revisionpurple}{RGB}{75,0,130} 
\definecolor{revisioncyan}{RGB}{0,160,190} 
\definecolor{darkgreen}{rgb}{0,0.5,0}
\usepackage{tkz-fct}
\usepackage{tkz-euclide}
\usetikzlibrary{external}
\usetikzlibrary{shapes.geometric}
\usetikzlibrary{decorations.shapes,arrows,calc,decorations.markings,shapes,decorations.pathreplacing,shapes.arrows}
\usetikzlibrary{positioning}
\usetikzlibrary{topaths,calc}
\usetikzlibrary{arrows.meta,fit}
\tikzset{main node/.style={circle,fill=blue!20,draw,inner sep=1pt},}
\tikzset{highlight1/.style={rectangle,
                           fill=red!15,
                           rounded corners = 0.5 mm,
                           inner sep=1pt,
                           fit=#1}}

\newcommand\hlight[1]{\tikz[overlay, remember picture,baseline=-\the\dimexpr\fontdimen22\textfont2\relax]\node[rectangle,fill=blue!13,rounded corners,fill opacity = 0.2,draw,thick,text opacity =1] {$#1$};}

\newtheorem{theorem}{Theorem}[section]
\newtheorem{lemma}[theorem]{Lemma}
\newtheorem{proposition}[theorem]{Proposition}
\newtheorem{corollary}[theorem]{Corollary}
\newtheorem{sublemma}{Sublemma}

\theoremstyle{definition}
\newtheorem{definition}[theorem]{Definition}
\newtheorem{example}[theorem]{Example}
\newtheorem{remark}[theorem]{Remark}

\DeclareMathOperator*{\argmin}{arg\,min}
\newcommand{\cT}{\mathcal{T}}

\newcommand{\R}{\mathbb{R}}

\newcommand{\E}{\mathbb{E}}

\newcommand{\N}{\mathbb{N}}

\newcommand{\rd}{\mathrm{d}}
\newcommand{\V}{\mathbb{V}}

\DeclareMathOperator{\rank}{rank}

\newcommand{\cM}{\mathcal{M}}

\newcommand{\G}{\mathcal{G}}

\newcommand{\dep}{\operatorname{dep}}

\newcommand{\cR}{\mathcal{R}}

\newcommand{\supp}{\operatorname{supp}}
\DeclareMathOperator*{\lcm}{lcm}
\newcommand{\oprocendsymbol}{\hbox{$\diamond$}}
\newcommand{\oprocend}{\relax\ifmmode\else\unskip\hfill\fi\oprocendsymbol}

\title{Structural Averaged Controllability for Linear Ensemble Systems: Multi-input Case}

\author[A. N. Moghaddam]{Amirreza Neshaei Moghaddam}
\address{Department of Electrical and Computer Engineering\\
University of California, Los Angeles}
\email{amirnesha@ucla.edu}

\author[X. Chen]{Xudong Chen}
\address{Department of Electrical and Systems Engineering\\
Washington University in St. Louis}
\email{cxudong@wustl.edu}

\author[B. Gharesifard]{Bahman Gharesifard}
\address{Department of Mathematics and Statistics\\
Queen's University, Canada}
\email{bahman.gharesifard@queensu.edu}

\begin{document}

\begin{abstract}
We study structural averaged controllability for multi-input linear ensemble systems. In this problem, one asks whether a sparsity pattern admits a linear ensemble system that is averaged controllable. The single-input case has been characterized completely, while the general multi-input case has remained open. In this work, we give a complete characterization for the general multi-input case. In particular, we prove that in addition to accessibility, a necessary and sufficient condition for structural averaged controllability of ensemble systems is existence of a row-saturating matching for its \emph{core}, an acyclic subgraph associated with the system which plays a central role in the result for the single-input case. 
This leads to new edge constructive weight assignment algorithm in the multi-input case, which is substantially more delicate than in the single-input case.
\end{abstract}


\maketitle

\section{Introduction}\label{sec:intro}

Ensemble control studies families of systems indexed by a parameter and driven by a common input. This framework was motivated in part by quantum control, and has also become useful in problems involving uncertainty and large collections of agents~\cite{JSL-NK:06,RB-NK:00,XC:22}.

The main difficulty is that the control signal cannot depend on the parameter value of the individual systems. As the parameter space becomes richer, requiring simultaneous control of the whole ensemble becomes increasingly challenging. In particular, for multidimensional parameter spaces, strong forms of ensemble controllability are generally unavailable~\cite{XC:23}. This motivates weaker, yet still meaningful, control objectives. One such objective is averaged controllability, where the goal is to steer only the average state of the ensemble. For linear ensemble systems, averaged controllability is shown to be characterized by a rank condition involving a Kalman-type averaged controllability matrix~\cite{EZ:14}.

A structural version of this problem was introduced in~\cite{BG-XC:22}. The question investigated there is not whether a particular numerical choice of system matrices is averaged controllable, but whether a given sparsity pattern of these matrices admits at least one averaged controllable realization. This is in the spirit of structural systems theory, but the resulting notion is different from classical structural controllability introduced by Lin~\cite{CTL:74}. In fact, structural averaged controllability is strictly weaker than classical structural controllability~\cite{BG-XC:22}. This naturally leads to the problem of characterizing exactly which sparsity patterns are structurally averaged controllable.

The single-input case was resolved in~\cite{XC-BG:25} using the notion of the core of the directed graph. Roughly speaking, the core is the acyclic part that remains after removing the cyclic components and their successors. In the single-input case, structural averaged controllability is equivalent to the core containing a directed spanning path. 
The multi-input case is more subtle. A na\"{i}ve extension of the single-input spanning-path condition is no longer correct, because multiple input channels can contribute to the averaged controllability matrix in a more complicated manner.

In this paper, we give a necessary and sufficient condition for structural averaged controllability in the general multi-input case. As in the single-input setting, the condition depends only on the core of the graph, and the proof is constructive: we build a realization of the sparsity pattern that is averaged controllable. The construction first assigns edge weights on the core and is then extended to the remaining part of the graph. The main difference is that the multi-input problem requires a new edge-weight assignment algorithm on the core, rather than the simpler construction used in the single-input case.
The extension from the core to the full graph also becomes more involved. The edge weights produced by the multi-input algorithm do not have the simple form used in the single-input proof, so the extension requires a more careful argument.

The rest of the paper is organized as follows. In Section~\ref{sec:preliminaries}, we introduce the notation and formulate the problem. In Section~\ref{sec:main_result_and_proof}, we first prove a value-assignment result for the core and then state the main result for the whole graph, Theorem~\ref{thm:main}. We then prove the two directions of Theorem~\ref{thm:main}: necessity follows from the core condition, while sufficiency is proved by reducing to a simpler graph structure and extending the core construction to the full graph. We conclude in Section~\ref{sec:conclusion}.

\section{Preliminaries}\label{sec:preliminaries}

\subsection{Notation}

We let $\N:=\{0,1,2,\dotsc\}$. For a matrix $M$, its $(i,j)$-th entry is denoted by $M_{ij}$; when an index is a compound expression, we use a comma for clarity, e.g., $M_{i,j+1}$ or $M_{h,(k-1)m+i}$. For index sets $R$ and $C$, we write $M_{R,C}$ for the submatrix of $M$ formed by the rows indexed by $R$ and the columns indexed by $C$.

For a topological space $X$, we denote by $\mathcal B(X)$ its Borel $\sigma$-algebra, i.e., the smallest $\sigma$-algebra containing the open subsets of $X$. For a map $f:X\to Y$ and a set $Z\subseteq Y$, $f^{-1}(Z)$ denotes the preimage of $Z$ under $f$. If $f:X\to Y$ is measurable, we write
\[
\bm{\sigma}(f):=\{f^{-1}(Z):Z\in\mathcal B(Y)\}
\]
for the $\sigma$-algebra generated by $f$. Throughout the paper, the plain letter $\sigma$ denotes the ensemble parameter, while $\bm{\sigma}(f)$ denotes the $\sigma$-algebra generated by the map $f$.

For a measure space $(\Sigma,\mu)$ and a measurable set $\Omega\subseteq\Sigma$, we denote the indicator of $\Omega$ by ${\bf 1}_{\sigma\in\Omega}$. For an integrable function $g$, we write
\[
\E\big[g(\sigma){\bf 1}_{\sigma\in\Omega}\big]
:=\int_{\Omega}g(\sigma)\,d\mu.
\]
We use this indicator form by default. When it keeps an expression shorter, especially when another indicator already appears in the integrand, we also write
\[
\E_{\Omega}[g(\sigma)]
:=\int_{\Omega}g(\sigma)\,d\mu.
\]
Finally, $L^\infty(\Sigma;\R^{a\times b})$ denotes the space of bounded measurable functions from $\Sigma$ to $\R^{a\times b}$.

\subsection{Graph-theoretic notation}

A directed graph, or digraph, is a pair $G=(V,E)$, where $V$ is a finite set of nodes and $E\subseteq V\times V$ is a set of directed edges. We write $v_i v_j$ for the directed edge from $v_i$ to $v_j$. If $v_i v_j\in E$, then $v_i$ is an in-neighbor of $v_j$ and $v_j$ is an out-neighbor of $v_i$. For a subgraph $H$ of $G$, we define
\[
N_{\mathrm{out}}(H)
:=
\{v_j\in V:\text{ there exists }v_i\in H\text{ with }v_i v_j\in E\},
\]
and define $N_{\mathrm{in}}(H)$ analogously.

A walk from $v_i$ to $v_j$ is a finite sequence of nodes whose first node is $v_i$, whose last node is $v_j$, and whose consecutive nodes are connected by directed edges. The length of a walk $\tau$, denoted by $\ell(\tau)$, is the number of edges in it. In our notation, $\tau_{v_i v_j}$ is a walk from $v_i$ to $v_j$. Now given two walks $\tau_{v_i v_j}$ and $\tau_{v_j v_k}$, where the terminal node $v_j$ of $\tau_{v_i v_j}$ coincides with the initial node $v_j$ of $\tau_{v_j v_k}$, we can obtain a new walk $\tau_{v_i v_k} = \tau_{v_i v_j}\tau_{v_j v_k}$ by concatenating the two walks. A path is a walk with no repeated nodes. A path in a subgraph $H$ is maximal if it cannot be extended by adding one more outgoing edge in $H$ without repeating a node. A cycle is a closed walk with no repeated nodes except the initial and terminal node. A self-loop is considered as a cycle of length one. 
If $\tau$ is a cycle, then $\tau^r$ denotes the walk obtained by traversing $\tau$ exactly $r$ times, with $\tau^0$ interpreted as the empty walk.
If a path exists from $v_i$ to $v_j$, then $v_j$ is called a successor of $v_i$, and $v_i$ is called a predecessor of $v_j$. A digraph $G = (V,E)$ is called strongly connected if for any two distinct nodes $v_i,v_j \in V$, $v_i$ is both a predecessor and a successor of $v_j$.

\subsection{Problem formulation}

Let $\Sigma$ be a manifold, possibly with boundary, and let $\mu$ be a Borel probability measure on $\Sigma$. We assume that there exists a nonempty open set $\mathcal{O}\subseteq\operatorname{supp}\mu$ such that $\mu|_{\mathcal{O}}$ is absolutely continuous with respect to a smooth volume measure on $\mathcal{O}$.
We consider a continuum ensemble of linear time-invariant systems driven by a common $m$-dimensional control input:
\begin{equation}\label{eq:systemeq}
\frac{\partial}{\partial t} x(t,\sigma) = A(\sigma)x(t,\sigma) + B(\sigma) u(t),
\end{equation}
for all $\sigma\in \Sigma$, where $x(t,\sigma)\in \R^n$ denotes the state
corresponding to the parameter value~$\sigma$ at time~$t$, $u(t)\in \R^m$ is the common control input applied uniformly across
the ensemble, and $A:\Sigma \to \R^{n\times n}$ and $B:\Sigma\to \R^{n \times m}$ are bounded, measurable functions.

A control input is admissible if, for every $T>0$, the function $u:[0,T]\to\R^m$ is integrable. For each $t\ge 0$, we define the profile $\chi(t):\Sigma\to\R^n$ by $\chi(t)(\sigma):=x(t,\sigma)$, and we define the average state by
\[
\bar\chi(t):=\int_\Sigma \chi(t)(\sigma)\,d\mu
=
\int_\Sigma x(t,\sigma)\,d\mu.
\]
We have the following definition from~\cite{EZ:14}:
\begin{definition}\label{def:averagectrl}
System~\eqref{eq:systemeq} is {\bf averaged controllable} if, for any initial profile $\chi(0)\in \mathcal{L}^{\infty}(\Sigma; \R^n)$, any target average $x^*\in \R^n$, and any $T>0$, there exists an admissible control input $u(t)$ such that the solution of~\eqref{eq:systemeq} satisfies $\bar \chi(T) = x^*$.
\end{definition}

Since the system~\eqref{eq:systemeq} is determined by the $(A,B)$ pair, we simply say $(A,B)$ is averaged controllable if~\eqref{eq:systemeq} is. 
We now recall the following characterization from~\cite{EZ:14} which provides a necessary and sufficient condition for averaged controllability:

\begin{lemma}\label{lem:necsufforavectrl}
The ensemble system~\eqref{eq:systemeq} is averaged controllable if and only if the infinite block matrix
\begin{equation}\label{eq:defcab}
\mathcal{C}(A, B) := \left[ \int_\Sigma B \, \rd\mu,\; \int_\Sigma AB \, \rd\mu,\; \int_\Sigma A^2 B \, \rd\mu,\; \cdots \right]
\end{equation}
has rank $n$.
\end{lemma}

Similar to~\cite{XC-BG:25}, we deal with sparse $(A,B)$ pairs, where the pattern is represented by a digraph $G = (V,E)$ on $(n+m)$ nodes as follows:
\begin{enumerate}
    \item The node set $V$ is a disjoint union of two subsets: the $\alpha$-nodes
    \[
    V_\alpha = \{\alpha_1,\alpha_2,\dotsc,\alpha_n\},
    \]
    corresponding to the $n$ scalar states $x_i$, and the $\beta$-nodes
    \[
    V_\beta = \{\beta_1,\beta_2,\dotsc,\beta_m\},
    \]
    corresponding to the $m$ scalar inputs $u_i$.
    \item There is a directed edge from $\alpha_k$ to $\alpha_j$ if $A_{jk}$ is not identically zero, and a directed edge from $\beta_i$ to $\alpha_j$ if $B_{ji}$ is not identically zero. Moreover, the $\beta$-nodes do not have any in-neighbors.
\end{enumerate}
We refer to the graph $G$ defined above as the digraph induced by the pair $(A,B)$. We also denote by $\G_{n,m}$ the collection of all digraphs on $(n+m)$ nodes, with $n$ $\alpha$-nodes and $m$ $\beta$-nodes satisfying the condition that none of the $\beta$-nodes have any in-neighbors.

Given $G\in\G_{n,m}$, a pair $(A,B)$ is \textit{compliant} with $G$ if the digraph induced by $(A,B)$ is a subgraph of $G$. Equivalently, $A_{jk}$ is the zero function whenever $\alpha_k\alpha_j\notin E$, and $B_{ji}$ is the zero function whenever $\beta_i\alpha_j\notin E$.
For a graph $G \in \G_{n,m}$, we define
\begin{equation*}
\V(G) :=
\big\{(A,B)\in L^\infty(\Sigma;\R^{n\times n})
\times L^\infty(\Sigma;\R^{n\times m})
: (A,B)\text{ is compliant with }G\big\}.
\end{equation*}
We say that $G$ is structurally averaged controllable if there exists a pair $(A,B)\in\V(G)$ that is averaged controllable.

We next introduce the strong component decomposition of a digraph.

\begin{definition}\label{def:scd}
Let $G=(V,E)$ be a digraph. The strong component decomposition of $G$ is a node-set decomposition $V=\bigcup_{i=1}^N V_i$,
where the sets $V_i$ are pairwise disjoint, such that the following hold:
\begin{enumerate}
    \item For each $i\in\{1,\dotsc,N\}$, the subgraph $G_i$ induced by $V_i$ is strongly connected.
    \item If $\widehat G$ is any strongly connected subgraph of $G$, then $\widehat G$ is a subgraph of $G_i$ for some $i \in \{1,\dotsc,N\}$.
\end{enumerate}
We call these induced subgraphs the strong components of $G$.
\end{definition}

Note that a strong component may consist of a single node, with or without a self-loop. A strong component is called nontrivial if it contains a cycle. For example, each $\beta$-node for a graph $G \in \G_{n,m}$ is a trivial strong component. Now let $G_{\mathrm{cyc}}=(V_{\mathrm{cyc}},E_{\mathrm{cyc}})$ be the union of all nontrivial strong components of $G$. Moreover, let $V_{\mathrm{cyc}}^+$ be the set of successors of nodes in $V_{\mathrm{cyc}}$, including the nodes in $V_{\mathrm{cyc}}$ themselves, and let $G_{\mathrm{cyc}}^+$ be the subgraph of $G$ induced by $V_{\mathrm{cyc}}^+$. We are now in a position to formally introduce the core of a digraph.


\begin{definition}\label{def:core}
Given a digraph $G\in\G_{n,m}$, define $V^*:=V\setminus V_{\mathrm{cyc}}^+$,
and let $G^*$ be the subgraph of $G$ induced by $V^*$. We refer to $G^*$ as the core of $G$.
\end{definition}

By construction, a node belongs to $G^*$ if and only if it is not a successor of any cycle in $G$. In particular, $G^*$ is acyclic and contains all $\beta$-nodes. We also write 
\[
V_\alpha^*:=V^*\cap V_\alpha
\]
for the set of $\alpha$-nodes in the core, and let $n^* := |V_\alpha^*|$ denote the number of such nodes. Moreover, for $G \in \G_{n,m}$, we say $G$ is accessible if for each $\alpha$-node $\alpha_j$, there exists a $\beta$-node $\beta_i$ with a path from $\beta_i$ to $\alpha_j$.

A matrix $\mathcal{S}\in\{0,\ast\}^{r\times s}$ is a \textit{sparsity pattern} matrix, where $\mathcal{S}_{ij}=0$ indicates a fixed zero, and $\mathcal{S}_{ij}=\ast$ indicates an independent free parameter.
A realization of $\mathcal{S}$ is any $M\in\R^{r\times s}$ obtained by assigning a real value to every $\ast$ while keeping the zeros. We denote the set of all such realizations of $\mathcal{S}$ by
\[
\mathcal{P}(\mathcal{S}) := \{ M\in\mathbb{R}^{r\times s} : M_{ij}=0 \; \text{whenever} \; \mathcal{S}_{ij}=0\}.
\]

Under the convention used in the definition of $\V(G)$, a pair $(A,B)\in\V(G^*)$ has $A:\Sigma\to\R^{n^*\times n^*}$ and $B:\Sigma\to\R^{n^*\times m}$ after fixing an ordering of the $\alpha$-nodes in $G^*$. When $n^*=0$, we interpret the core matrices as having zero rows. Without loss of generality, we assume that the $\alpha$-nodes in $V^*$ are labeled as $\{\alpha_1, \dotsc, \alpha_{n^*}\}$ until Subsection~\ref{subsec:reduced_graphs} of the paper.

Since $G^*$ contains no cycles, paths in $G^*$ have length at most $n^*$. For $(A,B)\in\V(G^*)$, we define the truncated averaged controllability matrix
\[
\mathcal C^{\mathrm{tr}}(A,B)
:=
\left[
\int_\Sigma B\,\rd\mu,\,
\int_\Sigma AB\,\rd\mu,\,
\dotsc,\,
\int_\Sigma A^{n^*-1}B\,\rd\mu
\right].
\]
We let $\mathcal{S}_{G^*}\in\{0,\ast\}^{n^*\times mn^*}$ denote the sparsity pattern of this truncated averaged controllability matrix. Specifically, for $h,k\in\{1,\dotsc,n^*\}$ and $i\in\{1,\dotsc,m\}$,
\[
(\mathcal{S}_{G^*})_{h,(k-1)m+i} =
\begin{cases}
\ast, & \text{if there exists a path of length }k\text{ from }\beta_i\text{ to } \alpha_h \text{ in } G^*,\\
0, & \text{otherwise.}
\end{cases}
\]
Thus, $\mathcal{P}(\mathcal{S}_{G^*})$ consists of matrices that can only take nonzero values in entries corresponding to paths in the core.

\section{Main result and proof}\label{sec:main_result_and_proof}
We begin with a constructive result on the core.

\subsection{Core controllability matrix assignment}
With the required definitions in place, we first prove the following value assignment theorem for the core.

\begin{theorem}\label{thm:controllability_matrix_value_assignment}
For any target matrix $M \in \mathcal{P}(\mathcal{S}_{G^*})$, there exists a pair $(A,B) \in \V(G^*)$ such that
\[
\mathcal C^{\mathrm{tr}}(A,B)=M.
\]
\end{theorem}

We prove Theorem~\ref{thm:controllability_matrix_value_assignment} by construction through Algorithm~\ref{alg: core_weight_assignment}. 
Note that the entry $M_{h,(k-1)m+i}$ corresponds to paths of length $k$ from $\beta_i$ to $\alpha_h$. Moreover, the auxiliary matrix $M'$ records the discrepancy between the target entry and the value already assigned to the corresponding path at the current stage of the algorithm. This discrepancy is updated in lines~\ref{alg_line: M'_update_ifnoteqzero} and~\ref{alg_line: M'_update_ifeqzero} of Algorithm~\ref{alg: core_weight_assignment}.

We say that $D_1,\dotsc,D_q$ form an \emph{equal-measure partition} of $D'$ if they are pairwise disjoint, their union is $D'$,
and
\[
\mu(D_1)=\cdots=\mu(D_q)=\frac{\mu(D')}{q}.
\]
Since $\mu|_{\mathcal O}$ is absolutely continuous with respect to a smooth volume measure, it is nonatomic. Hence every measurable $D'\subseteq\mathcal O$ of positive measure admits the finite equal-measure partitions required in Algorithm~\ref{alg: core_weight_assignment}; see, e.g.,~\cite{VIB:07}.
In the algorithm, the active set at a given step is the measurable subset of the parameter space $\Sigma$ on which that step of the edge-weight assignment acts.

\begin{algorithm}[!t]
\caption{Core Edge Weight Assignment}
\begin{algorithmic}[1]\label{alg: core_weight_assignment}
\STATE \textbf{Input:} target truncated averaged controllability matrix $M$ and a measurable active set $D\subseteq\mathcal O$ with $\mu(D)>0$.
\STATE \textbf{Initialize:} $M' \leftarrow M$, and set every edge-weight function equal to zero on all of $\Sigma$.
\STATE Partition $D$ into $m$ disjoint sets of equal measure, denoted by $\Sigma_1, \dotsc, \Sigma_m$, one for each collection of paths starting from input $\beta_i$.
\FOR{$i=1$ to $m$}
    \STATE Separate all unique maximal paths in $G^*$ starting from $\beta_i$, denoted each by $p_{ij}$, where $j \in \{1,2,\dotsc,s_i\}$ and $s_i$ is the number of such paths starting from $\beta_i$, and partition $\Sigma_i$ into $s_i$ disjoint sets $\Sigma_{ij}$ of equal measure. We let $\xi_{ij}(k)$ denote the index of the $k$-th $\alpha$-node in path $p_{ij}$, and $\delta_{ij}$ denote the length of $p_{ij}$.
    \FOR{$j=1$ to $s_i$}
    \STATE Consider $p_{ij}=\beta_i \alpha_{\xi_{ij}(1)} \alpha_{\xi_{ij}(2)} \cdots \alpha_{\xi_{ij}(\delta_{ij})}$.
        \STATE \textbf{Initialize:} $\tilde{r}_{ij}^{(0)} \leftarrow \mu(\Sigma_{ij})=\frac{\mu(D)}{m s_i}$ and $\tilde{\Sigma}^{(0)}_{ij} \leftarrow \Sigma_{ij}$, where $\tilde{\Sigma}^{(k)}_{ij}$ denotes the active set at depth $k$ of path $p_{ij}$.
        \FOR{$k = 1$ to $\delta_{ij}$}
            \STATE Denote the corresponding edge-weight function by $\tilde{e}_{ij}^{(k)}$.
            \IF{$M'_{\xi_{ij}(k), (k-1)m+i} \neq 0$}
                \STATE Update 
                \[
                \tilde{e}_{ij}^{(k)}(\sigma) \leftarrow \begin{cases}
                    \frac{M'_{\xi_{ij}(k), (k-1)m+i}}{\tilde{r}_{ij}^{(k-1)}} & \text{if } \sigma \in \tilde{\Sigma}^{(k-1)}_{ij} \\
                    \tilde{e}_{ij}^{(k)}(\sigma) & \text{otherwise}
                \end{cases}.
                \] \label{alg_line: edge_update_ifnoteqzero}
                \STATE Update $\tilde{r}_{ij}^{(k)} \leftarrow M'_{\xi_{ij}(k), (k-1)m+i}$. \label{alg_line: r_update_M'_nonzero}
                \STATE Keep $\tilde{\Sigma}^{(k)}_{ij} \leftarrow \tilde{\Sigma}^{(k-1)}_{ij}$.
                \STATE Update $M'_{\xi_{ij}(k), (k-1)m+i} \leftarrow 0$. \label{alg_line: M'_update_ifnoteqzero}             
            \ELSE
                \STATE Partition $\tilde{\Sigma}^{(k-1)}_{ij}$ into two sets $\tilde{\Sigma}^{(k-1,1)}_{ij}$ and $\tilde{\Sigma}^{(k-1,2)}_{ij}$ of equal measure. \label{alg_line: equal_measure_partition}
                \STATE Update
                \[
                \tilde{e}_{ij}^{(k)}(\sigma) \leftarrow \begin{cases}
                    2 & \text{if } \sigma \in \tilde{\Sigma}^{(k-1,1)}_{ij} \\[1mm]
                    -2 & \text{if } \sigma \in \tilde{\Sigma}^{(k-1,2)}_{ij} \\[0.5mm]
                    \tilde{e}_{ij}^{(k)}(\sigma) & \text{otherwise}
                \end{cases}.
                \]
                \STATE Update $\tilde{\Sigma}^{(k)}_{ij} \leftarrow \tilde{\Sigma}^{(k-1,1)}_{ij}$.
                \STATE Keep $\tilde{r}_{ij}^{(k)} \leftarrow \tilde{r}_{ij}^{(k-1)}$.
                \STATE Keep $M'_{\xi_{ij}(k), (k-1)m+i} \leftarrow 0$. \label{alg_line: M'_update_ifeqzero}  
            \ENDIF
        \ENDFOR
    \ENDFOR
\ENDFOR
\RETURN $(A,B)$.
\end{algorithmic}
\end{algorithm}

Moreover, note that all updates in Algorithm~\ref{alg: core_weight_assignment} are supported on subsets of $D$. Since every edge-weight function, which corresponds to values in matrices $A$ and $B$, is initialized to zero on $\Sigma$, the returned pair $(A,B)$ vanishes on $\Sigma\setminus D$. Consequently, every integral over $\Sigma$ appearing in Theorem~\ref{thm:controllability_matrix_value_assignment} is equal to the corresponding integral over $D$.
Observe also that every quantity $\tilde r_{ij}^{(k)}$ appearing in a denominator in Algorithm~\ref{alg: core_weight_assignment} is nonzero. Indeed, this quantity is initialized at $\tilde r_{ij}^{(0)}=\mu(\Sigma_{ij})>0$, and thereafter, $\tilde r_{ij}^{(k)}$ is either assigned the nonzero value $M'_{\xi_{ij}(k),(k-1)m+i}$ or kept equal to $\tilde r_{ij}^{(k-1)}$.


\begin{proof}
    If $n^*=0$, then, under the convention above, both $\mathcal C^{\mathrm{tr}}(A,B)$ and $\mathcal S_{G^*}$ are $0\times 0$ matrices. Hence $\mathcal P(\mathcal S_{G^*})$ contains only the $0\times 0$ matrix, and the statement holds trivially. Thus, in the proof below, we assume $n^*>0$.
    
    We start with a sublemma that shows that at any depth $d$ of a path $p_{ij} = \beta_i \alpha_{\xi_{ij}(1)} \alpha_{\xi_{ij}(2)} \cdots \alpha_{\xi_{ij}(\delta_{ij})}$, the value $\tilde{r}_{ij}^{(k)}$ captures the expected product of the edges of the path up to depth $k$ over the path's corresponding subspace $\Sigma_{ij}$. Formally, 
    \begin{sublemma} \label{sublem: r_value_induc}
    For any given $i \in \{1,2,\dotsc,m\}$ and $j \in \{1,2,\dotsc,s_i\}$, it holds for all $k \in \{0,1,\dotsc,\delta_{ij}\}$ that
    \begin{equation*}
    \tilde{r}_{ij}^{(k)} = \E \left[\prod_{{l}=0}^k \tilde{e}_{ij}^{({l})} (\sigma) \; {\bf 1}_{\sigma \in \tilde{\Sigma}^{(k)}_{ij}} \right].
    \end{equation*}
    \end{sublemma}
    \begin{proof}[Proof of Sublemma~\ref{sublem: r_value_induc}]
    We prove this by induction. For notational convenience, define
    \[
    \tilde e_{ij}^{(0)}(\sigma):=1,\qquad \sigma\in\Sigma,
    \]
    and recall that $\tilde\Sigma_{ij}^{(0)}=\Sigma_{ij}$.

    \textbf{Base case ($k = 0$):} 
    \begin{align*}
        \E [{\bf 1}_{\sigma \in \Sigma_{ij}}] = \mu (\Sigma_{ij}) =  \frac{\mu (D)}{m s_i} = \tilde{r}_{ij}^{(0)}.
    \end{align*}

    \textbf{Inductive step:} Let $k' \in \{0,1,\dotsc,\delta_{ij}-1\}$ be fixed and assume that
    \begin{equation}
    \tilde{r}_{ij}^{(k')} = \E \left[\prod_{{l}=0}^{k'} \tilde{e}_{ij}^{({l})} (\sigma) \; {\bf 1}_{\sigma \in \tilde{\Sigma}^{(k')}_{ij}} \right]. \label{eq: induc_hyp}
    \end{equation}
    Note that by construction,
    \[
    \tilde\Sigma_{ij}^{(k')}\subseteq\tilde\Sigma_{ij}^{(k'-1)}\subseteq\cdots\subseteq\tilde\Sigma_{ij}^{(0)}=\Sigma_{ij},
    \]
    and for every ${l}\in\{0,1,\dotsc,k'\}$, the function $\tilde e_{ij}^{({l})}$ is constant on $\tilde\Sigma_{ij}^{(k')}$. In particular, at depth ${l}$, $\tilde e_{ij}^{({l})}$ is assigned a constant value on the current active set $\tilde\Sigma_{ij}^{(l)}$, and all later active sets are subsets of that set. Consequently, $\prod_{{l}=0}^{k'}\tilde e_{ij}^{({l})}$ is constant on $\tilde\Sigma_{ij}^{(k')}$.

    Now based on the algorithm, two cases may happen:
    
    \textbf{Case 1:} If $M'_{\xi_{ij}(k'+1),k' m+i} \neq 0$.
    \begin{align*}
        \E \left[\prod_{{l}=0}^{k'+1} \tilde{e}_{ij}^{({l})} (\sigma) {\bf 1}_{\sigma \in \tilde{\Sigma}^{(k'+1)}_{ij}}\right] &\overset{\mathrm{(i)}}{=} \E \left[\prod_{{l}=0}^{k'+1} \tilde{e}_{ij}^{({l})} (\sigma) {\bf 1}_{\sigma \in \tilde{\Sigma}^{(k')}_{ij}}\right] \\
        &= \E \left[\left( \prod_{{l}=0}^{k'} \tilde{e}_{ij}^{({l})} (\sigma) \right) \left( \tilde{e}_{ij}^{(k'+1)} (\sigma) {\bf 1}_{\sigma \in \tilde{\Sigma}^{(k')}_{ij}} \right) \right] \\
        &= \E \left[\left( \prod_{{l}=0}^{k'} \tilde{e}_{ij}^{({l})} (\sigma) \right) \left( \frac{M'_{\xi_{ij}(k'+1), k' m+i}}{\tilde{r}_{ij}^{(k')}} {\bf 1}_{\sigma \in \tilde{\Sigma}^{(k')}_{ij}} \right) \right] \\
        &= \frac{M'_{\xi_{ij}(k'+1), k' m+i}}{\tilde{r}_{ij}^{(k')}} \E \left[ \prod_{{l}=0}^{k'} \tilde{e}_{ij}^{({l})} (\sigma) {\bf 1}_{\sigma \in \tilde{\Sigma}^{(k')}_{ij}} \right] \\
        &\overset{\mathrm{(ii)}}{=} \frac{M'_{\xi_{ij}(k'+1), k' m+i}}{\tilde{r}_{ij}^{(k')}} \tilde{r}_{ij}^{(k')} \\
        &= M'_{\xi_{ij}(k'+1), k' m+i} \\
        &\overset{\mathrm{(iii)}}{=} \tilde{r}_{ij}^{(k'+1)},
    \end{align*}
    where (i) follows from the fact that $\tilde{\Sigma}^{(k'+1)}_{ij} = \tilde{\Sigma}^{(k')}_{ij}$ in this case, (ii) from the induction hypothesis~\eqref{eq: induc_hyp}, and (iii) from line~\ref{alg_line: r_update_M'_nonzero} of Algorithm~\ref{alg: core_weight_assignment}.

    \textbf{Case 2:} Suppose $M'_{\xi_{ij}(k'+1),k'm+i}=0$. By the constancy argument above, we know there exists a constant $c$ such that
    \[
    \prod_{{l}=0}^{k'}\tilde e_{ij}^{({l})}(\sigma) = c,
    \qquad \text{for all } \sigma\in\tilde\Sigma_{ij}^{(k')}.
    \]
    Since $\tilde\Sigma_{ij}^{(k'+1)}=\tilde\Sigma_{ij}^{(k',1)}$, and $\tilde e_{ij}^{(k'+1)} (\sigma)=2$ for $\sigma \in \tilde\Sigma_{ij}^{(k',1)}$, we obtain
    \begin{align*}
        \E \left[\prod_{{l}=0}^{k'+1} \tilde{e}_{ij}^{({l})} (\sigma) {\bf 1}_{\sigma \in \tilde{\Sigma}^{(k'+1)}_{ij}}\right] &= \;2\,\E \left[\prod_{{l}=0}^{k'} \tilde{e}_{ij}^{({l})} (\sigma) {\bf 1}_{\sigma \in \tilde{\Sigma}^{(k',1)}_{ij}}\right] \\
        &=\;2c\mu\big(\tilde\Sigma_{ij}^{(k',1)}\big) \\
        &\overset{\mathrm{(i)}}{=}\;c\mu\big(\tilde\Sigma_{ij}^{(k')}\big) \\
        &=\;\E \left[\prod_{{l}=0}^{k'} \tilde{e}_{ij}^{({l})} (\sigma) {\bf 1}_{\sigma \in \tilde{\Sigma}^{(k')}_{ij}}\right] \\
        &\overset{\mathrm{(ii)}}{=}\;\tilde r_{ij}^{(k')} \\
        &=\;\tilde r_{ij}^{(k'+1)},
    \end{align*}
    where (i) follows from line~\ref{alg_line: equal_measure_partition} of Algorithm~\ref{alg: core_weight_assignment}, and (ii) from the induction hypothesis~\eqref{eq: induc_hyp}.
    This finishes the proof of Sublemma~\ref{sublem: r_value_induc}. \end{proof}

    We now provide the following result regarding the updates of $M'$ values in the algorithm. For notational convenience, whenever $k>\delta_{ij}$, we set $\xi_{ij}(k):=\varnothing$, so that ${\bf 1}_{\xi_{ij}(k)=h}=0$ for every $h\in\{1,\dotsc,n^*\}$ and the corresponding summand is zero.

    \begin{sublemma}\label{sublem: M'_induc}
    For $i\in\{1,\dotsc,m\}$, $k\in\{1,\dotsc,n^*\}$, and $j\in\{0,1,\dotsc,s_i\}$, let $M'_{h,(k-1)m+i}(j)$ denote the value of $M'_{h,(k-1)m+i}$ after the loop corresponding to path $p_{ij}$ has been completed, with $j=0$ denoting its initialized value. Then, for every $h\in\{1,\dotsc,n^*\}$,
    \begin{align}
    M_{h,(k-1)m+i} -\!\sum_{j'=1}^{j} \E_{\Sigma_{ij'}}\!\left[ \left( \prod_{{l}=1}^{k}\tilde e_{ij'}^{({l})}(\sigma) \right) {\bf 1}_{\xi_{ij'}(k)=h} \right] = M'_{h,(k-1)m+i}(j), \label{eq: M'_induc_argument}
\end{align}
where for $j=0$, the sum is understood to be zero.
\end{sublemma}

\begin{proof}[Proof of Sublemma~\ref{sublem: M'_induc}]
We proceed by induction on $j$.

\textbf{Base case ($j=0$):}
Since the sum in~\eqref{eq: M'_induc_argument} is empty
and no path starting from $\beta_i$ has yet been processed,
\begin{equation*}
M'_{h,(k-1)m+i}(0)=M_{h,(k-1)m+i},
\end{equation*}
which proves the claim.

\textbf{Inductive step:}
Fix $\tilde j\in\{0,1,\dotsc,s_i-1\}$ and assume that
\eqref{eq: M'_induc_argument} holds with $j=\tilde j$.
We consider the next path $p_{i,\tilde j+1}$.

\textbf{Case 1:}
Suppose $\xi_{i,\tilde j+1}(k)\neq h$. This includes the case $k>\delta_{i,\tilde j+1}$ under the convention introduced earlier. If $k>\delta_{i,\tilde j+1}$, the loop for $p_{i,\tilde j+1}$ never reaches depth $k$; otherwise, its depth-$k$ update concerns a row different from $h$. Hence $M'_{h,(k-1)m+i}$ is not changed, i.e.,
\[
M'_{h,(k-1)m+i} (\tilde j + 1) = M'_{h,(k-1)m+i} (\tilde j).
\]
At the same time, the
new path contributes zero to the sum in
\eqref{eq: M'_induc_argument}. Therefore,
\begin{align*}
M_{h,(k-1)m+i}
-\!\sum_{j'=1}^{\tilde j + 1}
\E_{\Sigma_{ij'}}\!\left[
\left(
\prod_{{l}=1}^{k}\tilde e_{ij'}^{({l})}(\sigma)
\right)
{\bf 1}_{\xi_{ij'}(k)=h}
\right] &= M'_{h,(k-1)m+i}(\tilde j ) - 0 \\
&= M'_{h,(k-1)m+i}(\tilde j +1),
\end{align*}
and the induction claim follows.

\textbf{Case 2:}
Suppose $\xi_{i,\tilde j+1}(k)=h$ and
$M'_{h,(k-1)m+i}(\tilde j)\neq0$. By the induction
hypothesis, the left-hand side of
\eqref{eq: M'_induc_argument} at stage $\tilde j+1$
equals
\begin{align*}
&M'_{h,(k-1)m+i}(\tilde j)
-\E_{\Sigma_{i,\tilde j+1}}\!\left[
\prod_{{l}=1}^{k}
\tilde e_{i,\tilde j+1}^{({l})}(\sigma)
\right]
\\
\overset{\mathrm{(i)}}{={}}&
M'_{h,(k-1)m+i}(\tilde j)
-\E\!\left[
\left(
\prod_{{l}=0}^{k-1}
\tilde e_{i,\tilde j+1}^{({l})}(\sigma)
\right)
\tilde e_{i,\tilde j+1}^{(k)}(\sigma)
{\bf 1}_{\sigma\in
\tilde{\Sigma}_{i,\tilde j+1}^{(k-1)}}
\right]
\\
\overset{\mathrm{(ii)}}{=}{}&M'_{h,(k-1)m+i}(\tilde j) -\frac{M'_{h,(k-1)m+i}(\tilde j)}
{\tilde r_{i,\tilde j+1}^{(k-1)}}
\E\!\left[
\left(
\prod_{{l}=0}^{k-1}
\tilde e_{i,\tilde j+1}^{({l})}(\sigma)
\right)
{\bf 1}_{\sigma\in
\tilde{\Sigma}_{i,\tilde j+1}^{(k-1)}}
\right]
\\
\overset{\mathrm{(iii)}}{=} &\ 0 \overset{\mathrm{(iv)}}{=}
M'_{h,(k-1)m+i}(\tilde j+1),
\end{align*}
where (i) follows from the fact that, within $\Sigma_{i,\tilde j+1}$, the newly assigned $k$-th edge-weight function is zero outside $\tilde{\Sigma}_{i,\tilde j+1}^{(k-1)}$, (ii) from line~\ref{alg_line: edge_update_ifnoteqzero} of Algorithm~\ref{alg: core_weight_assignment}, (iii) from Sublemma~\ref{sublem: r_value_induc}, and (iv) from line~\ref{alg_line: M'_update_ifnoteqzero} of Algorithm~\ref{alg: core_weight_assignment}.

\textbf{Case 3:}
Suppose $\xi_{i,\tilde j+1}(k)=h$ and
$M'_{h,(k-1)m+i}(\tilde j)=0$. By the constancy
property established in the proof of
Sublemma~\ref{sublem: r_value_induc}, there is a
constant $c$ such that
\begin{equation*}
\prod_{{l}=0}^{k-1}
\tilde e_{i,\tilde j+1}^{({l})}(\sigma)=c,
\qquad
\text{for all } \sigma\in
\tilde{\Sigma}_{i,\tilde j+1}^{(k-1)}.
\end{equation*}
Consequently, the left-hand side of~\eqref{eq: M'_induc_argument} at stage $\tilde j+1$ can be written as
\begin{align*}
&M'_{h,(k-1)m+i} (\tilde{j}) - \E\left[\prod_{{l}=0}^k \tilde{e}_{i,\tilde{j}+1}^{({l})} (\sigma) \; {\bf 1}_{\sigma \in \tilde{\Sigma}^{(k-1)}_{i,\tilde{j}+1}} \right] \\
= &0 - \E\!\left[
\left(
\prod_{{l}=0}^{k-1}
\tilde e_{i,\tilde j+1}^{({l})}(\sigma)
\right)
\tilde e_{i,\tilde j+1}^{(k)}(\sigma)
{\bf 1}_{\sigma\in
\tilde{\Sigma}_{i,\tilde j+1}^{(k-1)}}
\right] \\
=
&- 2c\,\mu\!\left(
\tilde{\Sigma}_{i,\tilde j+1}^{(k-1,1)}
\right)
+
2c\,\mu\!\left(
\tilde{\Sigma}_{i,\tilde j+1}^{(k-1,2)}
\right) \\
\overset{\mathrm(i)}{=} &\ 0 \overset{\mathrm{(ii)}}{=} M'_{h,(k-1)m+i}(\tilde j+1),
\end{align*}
where (i) and (ii) follow from lines~\ref{alg_line: equal_measure_partition} and~\ref{alg_line: M'_update_ifeqzero} of Algorithm~\ref{alg: core_weight_assignment} respectively. This concludes the induction and proves
Sublemma~\ref{sublem: M'_induc}. 
\end{proof}

As an immediate consequence, we have the following result:
\begin{corollary}\label{cor: M'_is_zero_if_visited}
Fix $h\in\{1,\dotsc,n^*\}$, $k\in\{1,\dotsc,n^*\}$, and $i\in\{1,\dotsc,m\}$. If there exists $j_0\in\{1,\dotsc,s_i\}$ such that $\xi_{ij_0}(k)=h$, i.e., there exists a path of length $k$ from $\beta_i$ to $\alpha_h$, then
\[
M'_{h,(k-1)m+i}(j_0)=M'_{h,(k-1)m+i}(s_i)=0.
\]
Equivalently,
\begin{align}
    \sum_{j'=1}^{s_i}
    \E_{\Sigma_{ij'}}\!\left[
    \left(\prod_{{l}=1}^{k}\tilde e_{ij'}^{({l})}(\sigma)\right)
    {\bf 1}_{\xi_{ij'}(k)=h}
    \right]
    =M_{h,(k-1)m+i}.
    \label{eq: M'_argument_result}
\end{align}
\end{corollary}

Let $\tilde M$ denote the truncated averaged controllability matrix assigned at the end of the algorithm:
\begin{equation}
\tilde M:=\mathcal C^{\mathrm{tr}}(A,B)
=\left[\int_\Sigma B\,\rd\mu,\int_\Sigma AB\,\rd\mu,\dotsc,\int_\Sigma A^{n^*-1}B\,\rd\mu\right].
\label{eq: c_tilde_defintion}
\end{equation}
We now give the final preliminary result.

\begin{sublemma} \label{sublem: c_tilde_prefix_decomposition}
Consider the notation in Algorithm~\ref{alg: core_weight_assignment}. For every $i\in\{1,\dotsc,m\}$, $k\in\{1,\dotsc,n^*\}$, and $h\in\{1,\dotsc,n^*\}$, we have
\begin{align}
\tilde M_{h,(k-1)m+i}
=
\sum_{j'=1}^{s_i}
\E_{\Sigma_{ij'}}\!\left[
\left(\prod_{{l}=1}^{k}\tilde e_{ij'}^{({l})}(\sigma)\right)
{\bf 1}_{\xi_{ij'}(k)=h}
\right],
\label{eq: tildeM-prefix_2}
\end{align}
where $\tilde M$ is defined in~\eqref{eq: c_tilde_defintion}, and the edge weights on the right-hand side are the final values assigned by Algorithm~\ref{alg: core_weight_assignment}.
\end{sublemma}

\begin{proof}[Proof of Sublemma~\ref{sublem: c_tilde_prefix_decomposition}]
Fix $i\in\{1,\dotsc,m\}$ and $k\in\{1,\dotsc,n^*\}$. Let $\mathcal P_{i,k}$ be the set of all distinct length-$k$ paths in $G^*$ starting at $\beta_i$. For each $p\in\mathcal P_{i,k}$, let $\eta(p)$ denote the index of the last $\alpha$-node visited in $p$, and define
\[
J(p):=\{j'\in\{1,\dotsc,s_i\}:\text{the length-$k$ prefix of }p_{ij'}\text{ is }p\}.
\]
Note that for distinct paths $p,\bar p\in\mathcal P_{i,k}$, the sets $J(p)$ and $J(\bar p)$ are disjoint, since each maximal path has a unique length-$k$ prefix. We also define
\begin{equation}
\Omega(p):=\bigsqcup_{j'\in J(p)}\Sigma_{ij'}.
\label{eq:Omega_def}
\end{equation}
Moreover, for $p\in\mathcal P_{i,k}$, define
\[
E_p(\sigma):=\prod_{{l}=1}^{k}\tilde e_{ij'}^{({l})}(\sigma),
\]
for some $j'\in J(p)$. This is well-defined, since all maximal paths indexed by $J(p)$ share the same first $k$ edges.

By definition~\eqref{eq: c_tilde_defintion},
\[
\tilde M_{h,(k-1)m+i}
=
\int_\Sigma(A^{k-1}B)_{h,i}(\sigma)\,\rd\mu.
\]
Expanding $(A^{k-1}B)_{h,i}$ over all length-$k$ paths from $\beta_i$ to $\alpha_h$, we obtain
\begin{align}
(A^{k-1}B)_{h,i}(\sigma)
=
\sum_{\substack{p\in\mathcal P_{i,k}\\ \eta(p)=h}}
E_p(\sigma).
\label{eq:path_expansion_short}
\end{align}
Hence,
\begin{align}
\tilde M_{h,(k-1)m+i}
=
\sum_{\substack{p\in\mathcal P_{i,k}\\ \eta(p)=h}}
\int_\Sigma E_p(\sigma)\,\rd\mu.
\label{eq:prefix_step1}
\end{align}
Now note that Algorithm~\ref{alg: core_weight_assignment} updates the first $k$ edge-weight functions along $p$ only while processing maximal paths whose length-$k$ prefix is exactly $p$, and only on the corresponding sets $\Sigma_{ij'}$ with $j'\in J(p)$. Since all edge-weight functions are initialized at zero, it follows that
\[
E_p(\sigma)=0
\qquad
\text{for all }\sigma\notin\Omega(p).
\]
Therefore,
\[
E_p(\sigma)=E_p(\sigma){\bf 1}_{\sigma\in\Omega(p)},
\]
and~\eqref{eq:prefix_step1} becomes
\begin{align}
\tilde M_{h,(k-1)m+i}
=
\sum_{\substack{p\in\mathcal P_{i,k}\\ \eta(p)=h}}
\E\!\left[
E_p(\sigma){\bf 1}_{\sigma\in\Omega(p)}
\right].
\label{eq:prefix_step2}
\end{align}
Finally, since the sets $\Sigma_{ij'}$ in~\eqref{eq:Omega_def} are disjoint, it follows that for each $p\in\mathcal P_{i,k}$,
\begin{align}
\E\!\left[
E_p(\sigma){\bf 1}_{\sigma\in\Omega(p)}
\right] =
\sum_{j'\in J(p)}
\E_{\Sigma_{ij'}}[E_p(\sigma)] =
\sum_{j'\in J(p)}
\E_{\Sigma_{ij'}}\!\left[
\prod_{{l}=1}^{k}\tilde e_{ij'}^{({l})}(\sigma)
\right].
\label{eq:prefix_step3}
\end{align}
Substituting~\eqref{eq:prefix_step3} into~\eqref{eq:prefix_step2}, we get
\begin{align*}
\tilde M_{h,(k-1)m+i} =
\sum_{\substack{p\in\mathcal P_{i,k}\\ \eta(p)=h}}
\sum_{j'\in J(p)}
\E_{\Sigma_{ij'}}\!\left[
\prod_{{l}=1}^{k}\tilde e_{ij'}^{({l})}(\sigma)
\right] =
\sum_{j'=1}^{s_i}
\E_{\Sigma_{ij'}}\!\left[
\left(\prod_{{l}=1}^{k}\tilde e_{ij'}^{({l})}(\sigma)\right)
{\bf 1}_{\xi_{ij'}(k)=h}
\right],
\end{align*}
where the last equality follows from the disjointness of the sets $J(p)$ for $p\in\mathcal P_{i,k}$, together with the fact that $\eta(p)=h$ is equivalent to $\xi_{ij'}(k)=h$ for $j'\in J(p)$. This proves Sublemma~\ref{sublem: c_tilde_prefix_decomposition}. 
\end{proof}

\noindent\emph{Conclusion of the proof of Theorem~\ref{thm:controllability_matrix_value_assignment}.}
By Corollary~\ref{cor: M'_is_zero_if_visited} and Sublemma~\ref{sublem: c_tilde_prefix_decomposition}, for any $(h,k,i)$, if there exists a length-$k$ path from $\beta_i$ to $\alpha_h$, then
\[
\tilde M_{h,(k-1)m+i}=M_{h,(k-1)m+i}.
\]
Assume, toward a contradiction, that there is some $(h,k,i)$ with
\[
\tilde M_{h,(k-1)m+i}\neq M_{h,(k-1)m+i}.
\]
Then no length-$k$ path from $\beta_i$ to $\alpha_h$ can exist, as otherwise equality would hold by the previous sentence. But if no such path exists, the sparsity constraint $M\in\mathcal P(\mathcal S_{G^*})$
forces $M_{h,(k-1)m+i}=0$, and since all terms in the corresponding sum in~\eqref{eq: tildeM-prefix_2} vanish, Sublemma~\ref{sublem: c_tilde_prefix_decomposition} implies that $\tilde M_{h,(k-1)m+i}=0$ as well. This contradicts $\tilde M_{h,(k-1)m+i}\neq M_{h,(k-1)m+i}$. Therefore, $\tilde M_{h,(k-1)m+i}=M_{h,(k-1)m+i}$ for all $(h,k,i)$, i.e.,
\[
\tilde M=\mathcal C^{\mathrm{tr}}(A,B)=M,
\]
completing the proof.
\end{proof}
The following remark explains, through an example, why the disjoint-union form in~\eqref{eq: tildeM-prefix_2} is necessary:

\begin{remark}[Why the disjoint-union form is necessary]
The disjoint-sum in~\eqref{eq: tildeM-prefix_2} prevents \emph{shared-prefix double counting}.
Consider Fig.~\ref{fig: example_for_disjoint_union} with two maximal paths
\(
p_{11}: \beta_1 \alpha_1 \alpha_2 \alpha_3
\) and
\(
p_{12}: \beta_1 \alpha_1 \alpha_2 \alpha_4
\).
Let $\Sigma=\Sigma_{11}\cup\Sigma_{12}$ with
$\mu(\Sigma_{11})=\mu(\Sigma_{12})=\tfrac12$, and define the first two edge functions (common to both paths) by
\begin{align*}
&\tilde{e}^{(1)}_{11}=\tilde{e}^{(1)}_{12}=\beta_1 \alpha_1=
\begin{cases}
5,& \sigma\in\Sigma_{11}\\
0,& \sigma\in\Sigma_{12}
\end{cases} \quad \text{and} \\
&\tilde{e}^{(2)}_{11}=\tilde{e}^{(2)}_{12}=\alpha_1 \alpha_2=
\begin{cases}
1,& \sigma\in\Sigma_{11}\\
0,& \sigma\in\Sigma_{12}
\end{cases}.
\end{align*}
For $(i,k,h)=(1,2,2)$ (length-$2$ paths from $\beta_1$ to $\alpha_2$), the correct form in~\eqref{eq: tildeM-prefix_2} gives
\[
\sum_{j'=1}^{2}\E_{\Sigma_{1j'}}\!\left[
\Big(\prod_{{l}=1}^{2}\tilde{e}^{({l})}_{1j'}(\sigma)\Big)\,
\mathbf 1_{\{\xi_{1j'}(2)=2\}}
\right]
=\tfrac12\,(5\cdot 1)+\tfrac12\,(0)=\tfrac52.
\]
If instead we summed expectations over all of $\Sigma$ as in
\[
\sum_{j'=1}^{2}\E\!\left[
\Big(\prod_{{l}=1}^{2}\tilde{e}^{({l})}_{1j'}(\sigma)\Big)\,
\mathbf 1_{\{\xi_{1j'}(2)=2\}}
\right],
\]
the nonzero contribution on $\Sigma_{11}$ would appear in \emph{both} summands (because the length-$2$ prefixes of $p_{11}$ and $p_{12}$ coincide), yielding
\(
\tfrac12(5\cdot 1)+\tfrac12(5\cdot 1)=5
\),
i.e., the shared prefix $\beta_1\!\to\!\alpha_1\!\to\!\alpha_2$ is counted twice. Restricting each term to the disjoint set $\Sigma_{1j'}$ assigned to the corresponding path avoids this duplication. \oprocend
\end{remark}
\begin{figure} 
\begin{center}
\begin{tikzpicture}[scale=1.0,node distance=1.4cm]   
\node[main node] (2)  {$\alpha_2$};  
\node[main node] (1) [left of= 2] {$\alpha_1$};
\node[main node] (b) [left of= 1] {$\beta_1$};
    \node[main node] (3) [above right of=2]  {$\alpha_3$};    
    \node[main node] (4) [below right  of=2]  {$\alpha_4$};
         
    \path[draw,shorten >=2pt,shorten <=2pt]    
    (b) edge[-latex] (1)
    (1) edge[-latex]  (2)
    (2) edge[-latex]  (3)
    (2) edge[-latex]  (4);
    \end{tikzpicture}
\end{center}
\caption{An intuitive example for Sublemma~\ref{sublem: c_tilde_prefix_decomposition}.} \label{fig: example_for_disjoint_union}
\end{figure}
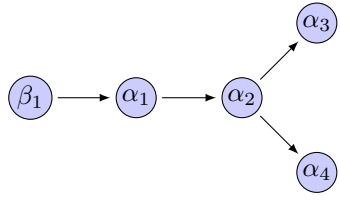

We provide the following definition, which associates each sparsity pattern matrix with a bipartite graph over its rows and columns.

\begin{definition}
Let $\mathcal{S}\in\{0,\ast\}^{m\times n}$ be a sparsity pattern matrix.
Define its \emph{row-column bipartite graph}
\[
\mathscr{B}(\mathcal{S}) = (V_R \cup V_C,\; E(\mathcal{S})),
\]
where $V_R=\left\{v_r^{(1)},\dots,v_r^{(m)}\right\}$ and $V_C=\left\{v_c^{(1)},\dots,v_c^{(n)}\right\}$ denote the nodes corresponding to the rows and columns of $\mathcal{S}$, respectively, and
\[
E(\mathcal{S}) = \{\, (v_r^{(i)},v_c^{(j)})\;:\; \mathcal{S}_{ij}=\ast \,\},
\]
i.e., there exists an edge joining $v_r^{(i)}$ and $v_c^{(j)}$ if and only if the $(i,j)$-th entry of $\mathcal{S} $ is nonzero.
Moreover, a matching $\cM \subseteq E(\mathcal{S})$ is called \emph{row-saturating} if every $v_r^{(i)}\in V_R$ is incident to exactly one edge in $\cM$.
\end{definition}

We now introduce the following intermediate result which provides an equivalent condition on when there exists a value assignment for a matrix with a certain sparsity pattern to make it full rank:

\begin{proposition} \label{prop: sparsity-matching}
    A sparsity pattern matrix $\mathcal{S}\in\{0,\ast\}^{m\times n}$ admits a realization of full row rank if and only if its row-column bipartite graph $\mathscr{B}(\mathcal{S})$ contains a row-saturating matching.
\end{proposition}
The proof of this result can be found in~\cite[Lemma~3]{MAB-XC-DZ:23}.

With Proposition~\ref{prop: sparsity-matching} established, combining it with Theorem~\ref{thm:controllability_matrix_value_assignment} yields the following:

\begin{corollary} \label{cor: full_rank_equiv_matching}
    $\mathcal{S}_{G^*}$ can be made full row rank if and only if $\mathscr{B}(\mathcal{S}_{G^*})$ has a row-saturating matching.
\end{corollary}

We now state the main result:
\begin{theorem} \label{thm:main}
    Let $G \in \G_{n,m}$ and $G^*$ be its core. Then $G$ is structurally averaged controllable if and only if $G$ is accessible and $\mathscr{B}(\mathcal{S}_{G^*})$ has a row-saturating matching.
\end{theorem}

Note that when $n^*=0$, the bipartite graph $\mathscr{B}(\mathcal{S}_{G^*})$ has no row nodes, so the row-saturating matching condition is trivially satisfied. In this case, Theorem~\ref{thm:main} says that $G$ is structurally averaged controllable if and only if $G$ is accessible.

Before proving the theorem, we first point out that the matching condition recovers the single-input condition from~\cite{XC-BG:25}.

\begin{corollary}\label{cor:single_input_reduction}
Suppose $m=1$. Then $\mathscr B(\mathcal S_{G^*})$ has a row-saturating matching if and only if $G^*$ contains a directed spanning path.
\end{corollary}
\begin{proof}
Recall that $n^*=|V_\alpha^*|$. When $m=1$, the matrix $\mathcal S_{G^*}$ has $n^*$ rows and $n^*$ columns, and its $(h,k)$-th entry is $\ast$ exactly when there is a path of length $k$ from $\beta_1$ to $\alpha_h$ in $G^*$.
First suppose that $G^*$ contains a directed spanning path 
$\beta_1\alpha_{h_1}\alpha_{h_2}\cdots \alpha_{h_{n^*}}$
visiting all $\alpha$-nodes in $G^*$. Then, for each $k\in\{1,\dotsc,n^*\}$, there is a path of length $k$ from $\beta_1$ to $\alpha_{h_k}$. Hence the entries
\[
(h_1,1),\quad
(h_2,2),\quad
\dotsc,\quad
(h_{n^*},n^*)
\]
are starred entries of $\mathcal S_{G^*}$, and they lie in distinct rows and columns. Thus they define a row-saturating matching in $\mathscr B(\mathcal S_{G^*})$.

Conversely, suppose $\mathscr B(\mathcal S_{G^*})$ has a row-saturating matching. Since $\mathcal S_{G^*}$ is an $n^*\times n^*$ matrix in the single-input case, the matching selects one starred entry in every row and every column. In particular, it selects some starred entry in column $n^*$. Therefore, there exists an $\alpha$-node $\alpha_h$ in $G^*$ that is reachable from $\beta_1$ by a path of length $n^*$. Since $G^*$ is acyclic and contains exactly $n^*$ $\alpha$-nodes, such a path must visit all $\alpha$-nodes in $G^*$. This gives the desired directed spanning path.
\end{proof}

We next provide an illustrative example showing that a na\"ive generalization of the spanning path condition in the single-input case fails for its multi-input counterpart.
A natural extension of the single-input characterization would be to ask for $m$ directed paths, one starting from each $\beta$-node, such that the union of the nodes visited in these paths covers all $\alpha$-nodes. 
While this condition is clearly sufficient, it is not necessary, as shown in the following example.

\begin{figure}[t]
\begin{center}
\begin{tikzpicture}[scale=1.0,node distance=1.6cm]
\node[main node] (b1) {$\beta_1$};
\node[main node] (b2) [below of=b1] {$\beta_2$};

\node[main node] (a1) [right of=b1, xshift=1.0cm] {$\alpha_1$};
\node[main node] (a2) [below of=a1] {$\alpha_2$};
\node[main node] (a3) [above right of=a2, xshift=0.8cm] {$\alpha_3$};
\node[main node] (a4) [below right of=a2, xshift=0.8cm] {$\alpha_4$};

\path[draw,shorten >=2pt,shorten <=2pt]
(b1) edge[-latex] (a1)
(b1) edge[-latex] (a2)
(b2) edge[-latex] (a2)
(a2) edge[-latex] (a3)
(a2) edge[-latex] (a4);
\end{tikzpicture}
\end{center}
\caption{An illustrative acyclic multi-input graph.}
\label{fig:multi_input_counterexample}
\end{figure}
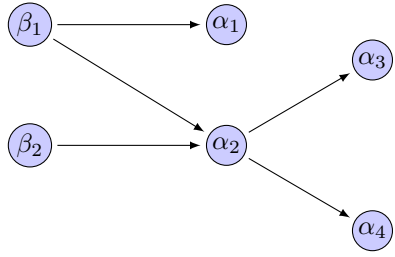

\begin{example}\label{ex:spanning_path_fail}
Consider the acyclic graph in Fig.~\ref{fig:multi_input_counterexample}. Since this graph is acyclic, its core is the graph itself.
Any path starting from $\beta_2$ can pass through $\alpha_2$ and then reach at most one of $\alpha_3$ and $\alpha_4$. Similarly, any path starting from $\beta_1$ either reaches only $\alpha_1$, or passes through $\alpha_2$ and then reaches at most one of $\alpha_3$ and $\alpha_4$. Hence, no choice of one path from $\beta_1$ and one path from $\beta_2$ can cover all four $\alpha$-nodes.

However, the graph satisfies the condition in Theorem~\ref{thm:main}. Since $G^*=G$ in this example, the core sparsity pattern is
\[
\mathcal{S}_{G^*} =
\left[
\begin{array}{cc|cc|cc|cc}
* & 0 & 0 & 0 & 0 & 0 & 0 & 0\\
* & * & 0 & 0 & 0 & 0 & 0 & 0\\
0 & 0 & * & * & 0 & 0 & 0 & 0\\
0 & 0 & * & * & 0 & 0 & 0 & 0
\end{array}
\right],
\]
where the block columns correspond to $B,AB,A^2B,A^3B$. The last two block columns are structurally zero because the graph has no directed paths of length greater than two. Now note that in the bipartite graph $\mathscr{B}(\mathcal{S}_{G^*})$, the starred entries
\[
(1,1),\qquad (2,2),\qquad (3,3),\qquad (4,4)
\]
lie in distinct rows and distinct columns. Thus, they form a row-saturating matching. By Theorem~\ref{thm:main}, the graph is structurally averaged controllable.
\oprocend
\end{example}

We now turn to the proof of Theorem~\ref{thm:main}.

\begin{proof}
    We first prove the necessity part, and leave the sufficiency for later. First, we show that if $G$ is not accessible, then it is not structurally averaged controllable. Since $G$ is not accessible, there exists at least one node $\alpha_i$ which is not a successor of any $\beta$-node. Let us label all such $\alpha$-nodes as $\alpha_1, \dotsc, \alpha_k$ for some $1 \le k \le n$. Then for any pair $(A,B) \in \mathbb{V} (G)$, we have
    \[
    A = \begin{bmatrix}
        A_{11} & 0 \\ A_{12} & A_{22}
    \end{bmatrix} \quad \text{and} \quad B=\begin{bmatrix}
        0 \\ B_2
    \end{bmatrix},
    \]
    where $A_{11}$ is $k$-by-$k$, $A_{22}$ is $(n-k)$-by-$(n-k)$, and $B_2$ is $(n-k)$-by-$m$. If we now partition $x = \begin{bsmallmatrix}
        x_1 \\ x_2
    \end{bsmallmatrix}$, the dynamics for $x_1$ will be as follows:
    \[
    \dot{x_1} (t,\sigma) = A_{11} (\sigma) x_1 (t,\sigma),
    \]
    which are not affected by the control input $u(t)$, and hence, the $x_1$ portion will not be averaged controllable, which then results in the whole system not being averaged controllable.

    We now turn our attention to the necessity of $\mathscr{B}(\mathcal{S}_{G^*})$ having a row-saturating matching. Recall that $n^* = |V_\alpha^*|$ is the number of $\alpha$-nodes in $G^*$. We label the nodes in $G$ so that $\alpha_1, \alpha_2, \dotsc, \alpha_{n^*}$ are in $G^*$, and the rest in $G_{\text{cyc}}^+$. It should be clear that any node in $G$ is either in $G^*$ or is a successor of some node in $G^*$. Therefore, by partitioning $x = \begin{bsmallmatrix}
        x_1 \\ x_2
    \end{bsmallmatrix}$, with $x_1$ denoting the $n^*$ $\alpha$-nodes in $G^*$, we can write the dynamics in the lower block triangular form as follows:
    \[
    \begin{bmatrix}
        \dot{x_1} (t,\sigma) \\ \dot{x_2} (t,\sigma)
    \end{bmatrix} = \begin{bmatrix}
        A_{11} (\sigma) & 0 \\ A_{12} (\sigma) & A_{22} (\sigma)
    \end{bmatrix} \begin{bmatrix}
        x_1 (t,\sigma) \\ x_2 (t,\sigma)
    \end{bmatrix} + \begin{bmatrix}
        B_1 (\sigma) \\ B_2 (\sigma)
    \end{bmatrix} u(t).
    \]
    Note that, as expected, the dynamics of $x_1 (t,\sigma)$ does not depend on $x_2 (t,\sigma)$. Thus, if $(A,B)$ is averaged controllable, so is $(A_{11},B_1)$. Consequently, if $G$ is averaged controllable, so is $G^*$. As a result, the sparsity pattern of the truncated averaged controllability matrix associated with $G^*$, namely $\mathcal{S}_{G^*}$, should admit a full row rank realization. Finally, by Corollary~\ref{cor: full_rank_equiv_matching}, we have that $\mathscr{B} (\mathcal{S}_{G^*})$ has a row-saturating matching, concluding the necessity part of the proof.
\end{proof}

\subsection{Reduced Graphs} \label{subsec:reduced_graphs}
We now turn to the sufficiency proof, though the final argument will be completed in Subsection~\ref{subsec: proof_of_suff}. The first step is to replace $G$, by removing edges if necessary, with a subgraph that has a simpler structure while still satisfying the two graph-theoretic conditions in Theorem~\ref{thm:main}.
Following the definitions in~\cite{XC-BG:25}, we introduce a class of digraphs called \textit{reduced graphs} with the intention to show that any digraph satisfying the conditions of Theorem~\ref{thm:main} can be reduced to such a reduced graph via edge removal. Note that this definition will be different from the one in~\cite{XC-BG:25} due to the differences induced by the presence of multiple inputs.

We start by recalling from~\cite{XC-BG:25} the notion of skeleton graph $S$ for any arbitrary digraph $G$ as follows:
\begin{definition}\label{def: skeleton}
    Let $G = (V,E)$, and $G_i = (V_i, E_i)$, for $i = 1,\dotsc,N$ be the strong components of $G$. The skeleton graph of $G$, denoted by $S = (W, F)$, is a directed graph on $N$ nodes $w_1, \dotsc , w_N$ whose edge set $F$ is determined by the following rule: There exists a directed edge $w_i w_j$ if there exists a directed edge in $G$ from a node of $G_i$ to a node of $G_j$.
\end{definition}

Next, we reintroduce the map $\pi : G \to S$, defined by sending a node $v_{i'} \in V_i$ to the node $w_i$, and we reiterate that, as can be observed from Definition~\ref{def: skeleton}, the map $\pi$ is a graph homomorphism. Moreover, let $G^*$ be the core, and $S^* := \pi (G^*)$. As a result of Definitions~\ref{def:core} and \ref{def: skeleton}, it should be clear that $S^*$ is the core of $S$ and that $G^*$ and $S^*$ are isomorphic under $\pi |_{G^*}$. In addition, we define $S_{\mathrm{cyc}} := \pi(G_{\mathrm{cyc}})$, which is the subgraph containing the strongly connected components of $S$; since $S$ is a skeleton graph, these correspond precisely to the nodes with self-loops (see~Definition~\ref{def: skeleton}). Similarly, define $S_{\mathrm{cyc}}^{+} := \pi(G_{\mathrm{cyc}}^{+})$ as the subgraph of $S$ induced by all successors of nodes in $S_{\mathrm{cyc}}$, including the nodes in $S_{\mathrm{cyc}}$ themselves. Having established these definitions, we now introduce:
\begin{definition} \label{def: reduced}
    A digraph $G \in \G_{n,m}$ is \textbf{reduced} if it satisfies the following:
    \begin{enumerate}
        \item Let $S$ be the skeleton of $G$, and $S^*$ be the core of $S$. Then $\mathscr{B}(\mathcal{S}_{S^*})$ has a row-saturating matching.
        \item For each edge $w_i w_j$ of $S$, with $w_i \neq w_j$, the set $\pi^{-1}(w_i w_j)$ is a singleton.
        \item If $w_p \in S$ has a self-loop, it is only an out-neighbor of exactly one node in $S^*$. Moreover, for any such $w_p$, $G_p = \pi^{-1} (w_p)$ is a cycle.
        \item If $w \in S_{\text{cyc}}^+ \setminus S_{\text{cyc}}$, $w$ is only an out-neighbor of exactly one node in $S_{\text{cyc}}^+$.
    \end{enumerate}
\end{definition}

In particular, due to the definition of skeleton in Definition~\ref{def: skeleton} and item~4 in Definition~\ref{def: reduced}, we have the following:

\begin{lemma} \label{lem: directed_forest}
    Let $G$ be a reduced graph and $S$ be its skeleton graph. Then $S_{\text{cyc}}^+$, ignoring self-loops, is a directed forest. In particular, each connected component of $S_{\mathrm{cyc}}^{+}$ is a directed tree rooted at a unique node $w_p$ with a self-loop in $S$ component.
\end{lemma}
\begin{proof}
We first prove $S_{\text{cyc}}^+$ is a directed forest by showing that its underlying undirected graph contains no cycles other than self-loops. 

Suppose, towards a contradiction, that the underlying undirected graph of $S_{\text{cyc}}^+$ contains a cycle involving at least two distinct nodes. Consider such a cycle. By items~3 and~4 of Definition~\ref{def: reduced}, each node in $S_{\text{cyc}}^+$ has at most one in-neighbor within $S_{\text{cyc}}^+$. Consequently, the only possible way to orient the edges of this undirected cycle is to form a directed cycle in $S_{\text{cyc}}^+$.

This leads to a contradiction. Indeed, since $S_{\text{cyc}}^+$ is a subgraph of the skeleton graph $S$, the existence of a directed cycle involving multiple nodes would imply that the corresponding strong components of $G$ are mutually reachable. By Definition~\ref{def: skeleton}, this would force these nodes to collapse into a single node in $S$, contradicting the assumption that the cycle involves distinct nodes.

Therefore, no such undirected cycle can exist, and the underlying undirected graph of $S_{\text{cyc}}^+$ is acyclic. Hence, $S_{\text{cyc}}^+$ is a directed forest when self-loops are ignored.

The final claim follows directly from the definition of $S_{\text{cyc}}^+$ together with items~3 and~4 of Definition~\ref{def: reduced}, which ensure that each node in $S_{\text{cyc}}^+ \setminus S_{\text{cyc}}$ has exactly one in-neighbor within $S_{\text{cyc}}^+$, while nodes with self-loops can only originate such components. Consequently, each connected component of $S_{\text{cyc}}^+$ contains a unique node with a self-loop, and this node serves as the root of the corresponding directed tree.
\end{proof}

In the single-input setting where $m=1$, Definition~\ref{def: reduced} coincides with the notion of a reduced graph introduced in~\cite{XC-BG:25}. To clarify, we have the following remark:
\begin{remark}
Our definition of a reduced digraph generalizes the single-input ($m=1$) notion of reduced graphs introduced in~\cite{XC-BG:25}. 
In the single-input case, the requirement that $\mathscr{B}(\mathcal{S}_{S^*})$ admit a row-saturating matching is equivalent to the core $S^*$ containing a directed spanning path, which is exactly the condition imposed in~\cite{XC-BG:25}. 
Moreover, when $m=1$, this spanning-path property, together with Lemma~\ref{lem: directed_forest}, implies that the skeleton graph $S$ is a directed tree, matching the structural characterization of reduced graphs in the single-input setting~\cite{XC-BG:25}.
\end{remark}

Recall that $G_{\text{cyc}}$ is the union of all nontrivial strong components of a graph $G$. Therefore, since $G$ is reduced, $G_{\text{cyc}}$ will just be a union of disjoint cycles. We denote these cycles by $G_p = (V_p, E_p)$, for $p \in \{1,2,\dotsc,q\}$, and let $\ell_p := \ell (G_p)$ for simplicity. We also let $V_p^+$ be the successors of $V_p$ in $G$, and $G_p^+$ the subgraph induced by $V_p^+$. Now by item~2 of Definition~\ref{def: reduced}, there is a unique node $\alpha_{p_0}$ in $G_p$ that is an out-neighbor of $G^*$:
\begin{equation}
    \{\alpha_{p_0}\} = N_{\text{out}} (G^*) \cap G_p. \label{eq: alpha_p_0_def}
\end{equation}
Recall from Lemma~\ref{lem: directed_forest} that every node in each connected component of $S_{\text{cyc}}^+$ is a successor of some node in $S$ with a self-loop. Now that we have established the notation, it should be clear that each such connected component of $S_{\text{cyc}}^+$ is just some $S_p^+ := \pi (G_p^+)$, and hence, $\alpha_{p_0}$ is a predecessor of all nodes in $G_p^+$. 
Furthermore, by item~3 of Definition~\ref{def: reduced}, we know that there is a unique node $v_{p^*}$ in $G^*$ that is the in-neighbor of $\alpha_{p_0}$, and consequently, the predecessor of all nodes in $G_p^+$. Note that the reason we denote this by $v_{p^*}$ is that this node can either be an $\alpha$-node or a $\beta$-node.
We now have the following on how reduced graphs satisfy conditions in Theorem~\ref{thm:main}:
\begin{lemma} \label{lem: reduced_satisfies_thm}
    Let $G$ be a reduced graph. Then $G$ satisfies the conditions in Theorem~\ref{thm:main}, i.e., $G$ is accessible and $\mathscr{B}(\mathcal{S}_{G^*})$ has a row-saturating matching.
\end{lemma}
\begin{proof}
    For the second claim, observe that Definition~\ref{def: reduced} enforces $\mathscr{B}(\mathcal{S}_{S^*})$ to admit a row-saturating matching. Now since $S^*$ and $G^*$ are isomorphic, as previously mentioned, this implies $\mathscr{B}(\mathcal{S}_{G^*})$ has a row-saturating matching as well, proving the second claim.
    
    Regarding the first claim, each $\alpha$-node $\alpha_k \in V^*$ is accessible, i.e., is a successor of some $\beta$-node, due to $\mathscr{B}(\mathcal{S}_{G^*})$ having a row-saturating matching. Moreover, for $\alpha$-nodes $\alpha_j \in V_{\text{cyc}}^+$, we know that $\alpha_j \in G_p^+$ for some $p \in \{1,2,\dotsc,q\}$. Let $v_{p^*}$ be as defined before. Since $v_{p^*} \in G^*$, it is either a successor of some $\beta$-node, or it is a $\beta$-node itself. Now combining this with the fact that $v_{p^*}$ is a predecessor of $\alpha_j$, we have that $\alpha_j$ is also a successor of that $\beta$-node, and hence, accessible. Finally, $V$ being the disjoint union of $V^*$ and $V_{\text{cyc}}^+$ concludes the first claim, finishing the proof.
\end{proof}

Having established this, we now provide the following result:
\begin{proposition} \label{prop: SAC_to_reduced}
    Let $G = (V,E) \in \G_{n,m}$ satisfy the conditions of Theorem~\ref{thm:main}. Then there is a subgraph $G' = (V,E')$ of $G$, with the same node set $V$, such that $G'$ is reduced.
\end{proposition}

The proof of Proposition~\ref{prop: SAC_to_reduced} follows exactly like~\cite[Proposition~4]{XC-BG:25} with the only difference that, in our case, there is no need to remove any edges on $G^*$, i.e., edges in $E^*$. Otherwise, the process of removing edges from nodes in $V^*$ to nodes in $V_{\text{cyc}}^+$, along with edges in $E_{\text{cyc}}^+$, remains the same.

\subsection{Enforcing Uniqueness on Walks to $G_{\text{cyc}}^+$} \label{subsec: uniqueness of walks to G_cyc_+}
We now aim to take advantage of the reduced definition to express all walks from $\beta$-nodes to the $\alpha$-nodes in $G_{\text{cyc}}^+$ in a unified formulation. We then switch focus to specifying one such walk for each $\alpha$-node in $G_{\text{cyc}}^+$, which will be used later in the final weight-assignment algorithm.

Let $G$ be a reduced graph, $\alpha_j \in V_p^+$, and $\alpha_{p_0}$ as defined before. We have already shown that there is a unique node $v_{p^*}$ in $G^*$ that is the in-neighbor of $\alpha_{p_0}$. Moreover, since $v_{p^*} \in G^*$, and $G$ is accessible, there exists some $\beta_i \in V_{\beta}$ that is the predecessor of $v_{p^*}$, or if $v_{p^*}$ is a $\beta$-node itself, we let $\beta_i = v_{p^*}$. Now as a result of Lemma~\ref{lem: directed_forest} and item~4 in Definition~\ref{def: reduced}, we have that every walk from $\beta_i$ to any $\alpha$-node $\alpha_j \in G_p^+$ has to be of the form:
\begin{equation}
\tau_{\beta_i \alpha_j} = \tau_{\beta_i v_{p^*}} \tau_{v_{p^*} \alpha_{p_0}} G_p^{k} \tau_{\alpha_{p_0} \alpha_j}, \label{eq: walks_to_cyc+}
\end{equation}
where, if $v_{p^*}$ is a $\beta$-node itself, i.e., $v_{p^*} = \beta_i$, then $\tau_{\beta_i v_{p^*}}$ is an empty walk. Here, $k \in \{0,1,2,\dotsc\}$ corresponds to the number of times the walk traverses the cycle $G_p$, and $k=0$ makes the walk in~\eqref{eq: walks_to_cyc+} a path. More importantly, $\tau_{v_{p^*} \alpha_{p_0}}$ is a unique edge due to item~2 in Definition~\ref{def: reduced}, and $\tau_{\alpha_{p_0} \alpha_j}$ is a unique path since $S_{\text{cyc}}^+$ is a directed forest as shown in Lemma~\ref{lem: directed_forest}. However $\tau_{\beta_i v_{p^*}}$ may not be unique. In particular, first of all, there may be multiple $\beta$-nodes that are predecessors of $v_{p^*}$. Secondly, there may be multiple paths of each of these $\beta$-nodes to $v_{p^*}$. However, in order to implement the same idea in~\cite{XC-BG:25} to make the rows corresponding to $G_{\mathrm{cyc}}^+$ contribute independent columns to the averaged controllability matrix, we need to enforce uniqueness of the relevant walks so that their contributions do not interfere with each other. To do so, let us first define
\begin{equation*}
    I_j := \{ i : \beta_i \text{ is a predecessor of } \alpha_j \},
\end{equation*}
and for some $i \in I_j$, let $\cT_{\beta_i \alpha_j}$ denote the set of all different paths from $\beta_i$ to $\alpha_j$. We now choose
\begin{equation*}
    i_j \in \argmin_{i \in I_j} \left( \min_{\tau \in \cT_{\beta_i \alpha_j}} \ell(\tau) \right),
\end{equation*}
and 
\begin{equation*}
    \tau^*_{\beta_{i_j} \alpha_j} \in \argmin_{\tau \in \cT_{\beta_{i_j} \alpha_j}} \ell(\tau).
\end{equation*}
If either minimizer is not unique, we choose one arbitrarily and fix it for the rest of the construction. Thus, $\beta_{i_j}$ is a $\beta$-node with the shortest path to $\alpha_j$, and that path is denoted by $\tau^*_{\beta_{i_j} \alpha_j}$. Now note that as a result of~\eqref{eq: walks_to_cyc+}, we have that for any $\alpha_j \in G_p^+$,
\begin{equation}
\tau^*_{\beta_{i_j} \alpha_j} = \tau^*_{\beta_{i_{p^*}} v_{p^*}} \tau_{v_{p^*} \alpha_{p_0}} \tau_{\alpha_{p_0} \alpha_j}, \label{eq: unique_shortest_path_to_cyc+}
\end{equation}
where, as previously mentioned, $\tau_{v_{p^*} \alpha_{p_0}}$ and $\tau_{\alpha_{p_0} \alpha_j}$ are unique paths. Moreover, we define the \textbf{depth} of each node $\alpha_j \in G$ as
\begin{equation}
    \dep (\alpha_j ) := \ell (\tau^*_{\beta_{i_j} \alpha_j}). \label{eq: depth_def}
\end{equation}

Now we have two possible scenarios:
\begin{enumerate}
    \item {\boldmath$v_{p^*}$ \textbf{is an $\alpha$-node}}: Then $\tau^*_{\beta_{i_{p^*}} v_{p^*}}$ is a non-empty path in $G^*$, and hence, it will be a part of some maximal path 
    \[
    p_{i_{p^*} j_{p^*}} = \beta_{i_{p^*}} \alpha_{\xi_{i_{p^*} j_{p^*}}(1)} \alpha_{\xi_{i_{p^*} j_{p^*}}(2)} \cdots \alpha_{\xi_{i_{p^*} j_{p^*}}(\delta_{i_{p^*} j_{p^*}})},
    \]
    where
    \[
    \xi_{i_{p^*} j_{p^*}}(\dep (v_{p^*})) = p^*,
    \]
    i.e., $v_{p^*}$ is the $\dep (v_{p^*})$-th node visited along this path. 
    \item {\boldmath$v_{p^*}$ \textbf{is a $\beta$-node}}: Then $\tau^*_{\beta_{i_{p^*}} v_{p^*}}$ is an empty path. This means $\tau^*_{\beta_{i_j} \alpha_j}$ has no edges in $G^*$. Note that this will play a role in the weight assignment later on.
\end{enumerate}

We have now derived a family of walks, parameterized by $k\in\N$, to each $\alpha$-node in $G_{\text{cyc}}^+$, namely, for $\alpha_j \in G_{p}^+$:
\begin{equation}
\tau_{\beta_{i_{p^*}} v_{p^*}}^* \tau_{v_{p^*} \alpha_{p_0}} G_p^{k} \tau_{\alpha_{p_0} \alpha_j}. \label{eq: unique_walks_to_cyc+}
\end{equation}
Later, in the weight-assignment algorithm, we will force the controllability matrix to only have values corresponding to these walks by letting the weight of the edges in $G_{\text{cyc}}^+$ only take nonzero values on the probability support of the unique path $\tau_{\beta_{i_{p^*}} v_{p^*}}^*$.

The scalar function $f$ constructed next is used to assign different powers to these uniquely characterized walks in $G_{\mathrm{cyc}}^+$, which later allows us to separate their contributions in the averaged controllability matrix.

\subsection{Construction of $f$}
We now construct the scalar function used in the extension of the core edge weight assignment. Let $d:=\dim\Sigma$. Choose a chart $(\mathcal{U},\phi)$ such that
\[
\mathcal U\subseteq\mathcal O\subseteq\operatorname{supp}\mu,
\]
where $\phi:\mathcal{U}\to\mathcal{V}$ is a diffeomorphism onto an open set $\mathcal V\subset\mathbb R^d$. After translating and rescaling the coordinates if necessary, we take $\mathcal V$ to contain the closed unit ball
\begin{equation}
\mathbb B := \{ z\in\mathbb{R}^d \mid \|z\|\le 1\},
\end{equation}
and let $K:=\phi^{-1}(\mathbb B)$. In particular, $K\subseteq\mathcal U\subseteq\mathcal O$ and $\mu(K)>0$. We then define the bounded measurable function $f:\Sigma\to[0,1]$ by
\begin{equation}\label{eq:fdef}
f(\sigma):=
\begin{cases}
\|\phi(\sigma)\|, & \text{if }\sigma\in K,\\
0, & \text{otherwise}.
\end{cases}
\end{equation}
Because $\mu|_{\mathcal O}$ is absolutely continuous with respect to a smooth volume measure, the pushforward measure $\phi_\#(\mu|_{\mathcal U})$ is absolutely continuous with respect to the $d$-dimensional Lebesgue measure on $\phi(\mathcal U)$.
\begin{lemma}\label{lem:radial_equal_measure_splitting}
For every $r\in[0,1]$,
\[
\mu\big(\{\sigma\in K:f(\sigma)=r\}\big)=0.
\]
Moreover, let $0\le a<b\le1$ and suppose that
\[
D_{a,b}:=\{\sigma\in K:a\le f(\sigma)\le b\}
\]
has positive measure. For every $\theta\in(0,1)$, there exists $c\in(a,b)$ such that
\[
\mu(D_{a,c})=\theta\mu(D_{a,b}),
\qquad
\mu(D_{c,b})=(1-\theta)\mu(D_{a,b}),
\]
where the common level set $f^{-1}(\{c\})\cap K$ may be assigned to either piece without changing its measure. Consequently, every positive-measure radial interval can be partitioned into any finite number of radial intervals of equal measure.
\end{lemma}
\begin{proof}
Let $\nu:=\phi_{\#}(\mu|_{\mathcal U})$. Then $\nu\ll\lambda_d$, where $\lambda_d$ is Lebesgue measure on $\mathbb R^d$. For $r>0$, the set $\{z\in\mathbb B:\|z\|=r\}$ is a Euclidean sphere, while for $r=0$ it is the singleton $\{0\}$. Each of these sets has $\lambda_d$-measure zero. Hence
\[
\mu\big(\{\sigma\in K:f(\sigma)=r\}\big)=0
\]
for every $r\in[0,1]$, proving the first claim. 

Now define
\[
F(t):=\mu\big(\{\sigma\in K:a\le f(\sigma)\le t\}\big),
\qquad t\in[a,b].
\]
Since $f_\#(\mu|_K)$ is nonatomic, $F$ has no jumps, and hence it is continuous. Also, $F(a)=0$ and $F(b)=\mu(D_{a,b})$. The intermediate value theorem therefore yields $c\in(a,b)$ with $F(c)=\theta\mu(D_{a,b})$. Since the level set at $c$ has measure zero, the resulting radial pieces $D_{a,c}$ and $D_{c,b}$ have the claimed measures. Repeating this argument gives an equal-measure partition into any finite number of radial intervals.
\end{proof}

\subsection{Extension of the Edge Weight Assignment}
First, recall the definition of $\dep (\alpha_i)$ from~\eqref{eq: depth_def}. We then relabel, if necessary, the nodes in $G$ such that
\[
\text{if} \ \dep(\alpha_i) < \dep(\alpha_j), \ \text{then} \ i < j.
\]
Similar to~\cite{XC-BG:25}, we now partition the edge set $E$ into five subsets $E_i$. Let $V_{\mathrm{appx}} := V_{\mathrm{cyc}}^{+} \setminus V_{\mathrm{cyc}}$, and
$G_{\mathrm{appx}} = (V_{\mathrm{appx}}, E_{\mathrm{appx}})$ be the subgraph induced by
$V_{\mathrm{appx}}$. We call $G_{\mathrm{appx}}$ the appendix of $G$.
Define the subsets $E_i$ of $E$ as follows:
\begin{equation}
\left\{
\begin{aligned}
E_1 &:= E^{*}, \\
E_2 &:= \{\, v_i \alpha_j \in E \mid v_i \in G^{*} \text{ and } \alpha_j \in G_{\mathrm{cyc}} \,\}, \\
E_3 &:= E_{\mathrm{cyc}}, \\
E_4 &:= \{\, \alpha_i \alpha_j \in E \mid \alpha_i \in G_{\mathrm{cyc}}
\text{ and } \alpha_j \in G_{\mathrm{appx}} \,\}, \\
E_5 &:= E_{\mathrm{appx}} .
\end{aligned}
\right.
\label{eq: edge_set_partition}
\end{equation}
The edges in $E_2$ link nodes from $G^{*}$ to $G_{\mathrm{cyc}}$, and edges in $E_4$ link nodes
from $G_{\mathrm{cyc}}$ to $G_{\mathrm{appx}}$. Moreover, for every $i \in \{1,\dotsc, m\}$, we define
\[
\ell_{\max}^{(i)} := \max\!\left(\{0\}\cup\{\ell_p:p\in\{1,\dotsc,q\},\ i_{p^*}=i\}\right),
\]
and let
\[
L := \lcm_{p=1}^{q} \ell_p,
\]
where $\lcm$ stands for least common multiple. Note that $\ell_{\max}^{(i)} = 0$ if there does not exist any $G_p^+$ such that $i_{p^*} = i$, i.e., there is no characterized walk of the form~\eqref{eq: unique_walks_to_cyc+} from $\beta_i$ to any $\alpha$-node in the $G_{\text{cyc}}^+$ portion of the graph.

With this notation in place, we now present the edge-weight assignment algorithm, on the whole graph $G$, that constructs a pair $(A,B)$  whose averaged controllability matrix has full row rank.

\begin{algorithm}
\caption{Total Weight Assignment}
\begin{algorithmic}[1]\label{alg: Whole_G_weight_assignment}
\STATE \textbf{Input:} the reduced graph $G\in\G_{n,m}$, the active set $K$, and the function $f$ constructed in~\eqref{eq:fdef}.
\STATE Since $G$ is reduced, there exists an $M\in\mathcal P(\mathcal S_{G^*})$ of full rank. Pick any such $M$ and run Algorithm~\ref{alg: core_weight_assignment} on $G^*$ with target truncated averaged controllability matrix $M$ and active set $D = K$. Choose all equal-measure partitions in this run to be radial, as depicted in Lemma~\ref{lem:radial_equal_measure_splitting}, and store the resulting $\tilde r_{ij}^{(k)}$ and $\tilde\Sigma_{ij}^{(k)}$ values.
\STATE Pick $\gamma$ to be a positive \textit{irrational} number that we fix throughout the algorithm.
\FOR{the remaining edges $e \in E \setminus E^*$}
    \IF{$e \in E_2$}
        \STATE $e$ can be shown as $e = v_{p^*} \alpha_{p_0}$. 
        \IF{$v_{p^*}$ is a $\beta$-node}
            \STATE $e \leftarrow f^{p_0 \gamma} \mathbf{1}_{\sigma \in K}.$ \label{alg_line: E2_if_beta}
        \ELSE
            \STATE $e \leftarrow \frac{\mu \left( \tilde{\Sigma}_{i_{p^*} j_{p^*}}^{(\dep (v_{p^*}))} \right)}{\tilde{r}_{i_{p^*} j_{p^*}}^{(\dep (v_{p^*}))}} f^{p_0 \gamma} \ {\bf 1}_{\sigma \in \tilde{\Sigma}_{i_{p^*} j_{p^*}}^{(\dep (v_{p^*}))}}.$ \label{alg_line: E2_if_alpha}
        \ENDIF
    \ELSIF{$e \in E_3$}
        \STATE Let $G_p$ be the cycle containing $e$. We write $G_p$ explicitly as $G_p = \alpha_{p_0} \alpha_{p_1} \dotsc \alpha_{p_{l_p - 1}} \alpha_{p_0}$.
        \STATE We then let
        \[
        e \leftarrow
        \begin{cases}
        f^{(\frac{\ell_p}{L})} & \text{if } e = \alpha_{p_{\ell_p-1}}\,\alpha_{p_0}, \\[6pt]
        1 & \text{otherwise}.
        \end{cases}
        \]
    \ELSIF{$e \in E_4$}
        \STATE Let $G_p$ be the cycle, $\alpha_j \in G_{\text{appx}}$ the node incident to $e$, and $\alpha_{p_0}$ as defined in~\eqref{eq: alpha_p_0_def}. We then set
        \[
        e \leftarrow f^{(j - p_0) \gamma}.
        \]
    \ELSE
        \STATE Then $e \in E_5$ and can be written as $e = \alpha_i \alpha_j$. Now set
        \[
        e \leftarrow f^{(j - i) \gamma}.
        \]
    \ENDIF
\ENDFOR
\RETURN $(A,B)$.
\end{algorithmic}
\end{algorithm}
We now aim to unify notation between lines~\ref{alg_line: E2_if_beta} and~\ref{alg_line: E2_if_alpha}. For some $G_p^+$ where $v_{p^*}$ is a $\beta$-node, we define
\[
\tilde{r}^{(\dep (v_{p^*}))}_{i_{p^*} j_{p^*}} := \mu (K) \quad \text{and} \quad \tilde{\Sigma}^{(\dep (v_{p^*}))}_{i_{p^*} j_{p^*}} := K.
\]
Note that this still matches the statement in Sublemma~\ref{sublem: r_value_induc} on the value of $\tilde{r}$ as
\[
\mu (K) = \E [ \mathbf{1}_{\sigma \in K}].
\]

Observe that, as promised in Subsection~\ref{subsec: uniqueness of walks to G_cyc_+}, line~\ref{alg_line: E2_if_alpha} of Algorithm~\ref{alg: Whole_G_weight_assignment} ensures that the walks to each $G_p^+$ take nonzero values only on $\tilde{\Sigma}_{i_{p^*} j_{p^*}}^{(\dep (v_{p^*}))}$, which is the probability support corresponding to the unique path $\tau_{\beta_{i_{p^*}} v_{p^*}}^*$. As a result, the only walks of nonzero value to any $\alpha$-node in $G_{\text{cyc}}^+$ will only be of the form in~\eqref{eq: unique_walks_to_cyc+}.
Similar to~\cite{XC-BG:25}, this allows us to establish a series of useful lemmas that leverage this uniquely characterized family of walks.

\subsection{Finalizing the Proof of Sufficiency for Theorem~\ref{thm:main}} \label{subsec: proof_of_suff}
For positive integers $i$ and $j$, let $V(i,j)$ be the set of $\alpha$-nodes $\alpha_h \in G$ such that there exists a walk $\tau$ from $\beta_i$ to $\alpha_h$ of length $j$, with the additional requirement that, whenever $\alpha_h \in G_p^+$ for some $p \in \{1,\dots,q\}$, the walk $\tau$ is of the form
\begin{equation*}
    \tau_{\beta_{i_{p^*}} v_{p^*}}^* \,
    \tau_{v_{p^*} \alpha_{p_0}} \,
    G_p^{k} \,
    \tau_{\alpha_{p_0}\alpha_h},
\end{equation*}
for $k \in \N$, which is the family of walks indicated in~\eqref{eq: unique_walks_to_cyc+}, and $i_{p^*} = i$. Note that, as previously mentioned, this extra restriction is due to the fact that, as a result of Algorithm~\ref{alg: Whole_G_weight_assignment}, the only walks of nonzero value to any $\alpha$-node in $G_{\text{cyc}}^+$ will only be of this form. No additional restriction is imposed when $\alpha_h \in G^*$. 
Moreover, for each $p \in \{1,\dots,q\}$, we let
\[
V_p^+ (i,j) := V_p^+ \cap V(i,j).
\]
Now recall that $V_{\text{cyc}}^+$ is the set of successors of all nontrivial strong components in $G$, and $G_{\text{cyc}}^+$ is the subgraph induced by $V_{\text{cyc}}^+$. Since $G$ is reduced,
\[
G_{\text{cyc}}^+ = \cup_{p=1}^{q} G_p^+.
\]
Also, define
\[
d_p := \max_{\alpha_j \in V_p^+} \dep (\alpha_j),
\]
and recall $n^* = |V_\alpha^*|$ is the number of $\alpha$-nodes in $G^*$.
We are now in a position to provide a series of results similar to~\cite{XC-BG:25} that will come in handy for the final proof. 
\begin{lemma} \label{lem: xc_bg_lemma5}
    Suppose that $j > n^*$ and $j > \max_{p=1}^{q} (d_p - \ell_p)$. Then
    \[
    V(i,j+L) = V(i,j) \quad \text{and} \quad V_{\text{cyc}}^+ = \cup_{i=1}^{m} \cup_{k=0}^{\ell_{\max}^{(i)} - 1} V(i,j+k).
    \]
\end{lemma}
The proof of Lemma~\ref{lem: xc_bg_lemma5} follows exactly like~\cite[Lemma~5]{XC-BG:25} as a result of the fact that for any $\alpha$-node in $G_{\text{cyc}}^+$, all the walks of nonzero value are of the form~\eqref{eq: unique_walks_to_cyc+}. Furthermore, note that since $j > n^*$, for each $i$, $V(i,j')$ with $j' \ge j$ only contains $\alpha$-nodes in $G_{\text{cyc}}^+$. Thus, by the definition of $V(i,j)$, we have that
\begin{equation}
\left( \cup_{k=0}^{\ell_{\max}^{(i)} - 1} V(i,j+k) \right) \cap \left( \cup_{k=0}^{\ell_{\max}^{(i')} - 1} V(i',j+k) \right) = \varnothing, \label{eq: disjoint_feeding_to_cyc_plus}
\end{equation}
for $i \neq i'$. This means that the subsets of $\alpha$-nodes in $G_{\mathrm{cyc}}^+$ reached from different $\beta$-nodes through the characterized walks in~\eqref{eq: unique_walks_to_cyc+} are completely disjoint.

Let $\alpha_h \in G_p^+$, and recall the walks $\tau$ of the form characterized in~\eqref{eq: unique_walks_to_cyc+}:
\[
\tau = \tau_{\beta_{i_{p^*}} v_{p^*}}^* \tau_{v_{p^*} \alpha_{p_0}} G_p^{k} \tau_{\alpha_{p_0} \alpha_h}.
\]
We use $\nu$ as a mapping from this set of walks to nonnegative real numbers such that
\begin{equation}
    \nu (\tau) := \begin{cases}
        p_0 \gamma + \frac{\ell (\tau) - \dep (\alpha_h)}{L} & \text{if } \alpha_h \in G_p, \\
        h \gamma + \frac{\ell (\tau) - \dep (\alpha_h)}{L} & \text{if } \alpha_h \in G_p^+ \setminus G_p.
    \end{cases} \label{eq: nu_def_on_walks}
\end{equation}

We now have the following result:
\begin{lemma} \label{lem: prod_tau2_formulation}
    Let $\alpha_h \in G_p^+$ and $\tau$ be a characterized walk from $\beta_{i_{p^*}}$ to $\alpha_h$, i.e.,
    \[
    \tau = \tau_{\beta_{i_{p^*}} v_{p^*}}^* \tau_{v_{p^*} \alpha_{p_0}} G_p^{k} \tau_{\alpha_{p_0} \alpha_h}.
    \]
    Let us denote
    \[
    \tau_{1} := \tau_{\beta_{i_{p^*}} v_{p^*}}^* \quad \text{and} \quad \tau_2 := \tau_{v_{p^*} \alpha_{p_0}} G_p^{k} \tau_{\alpha_{p_0} \alpha_h},
    \]
    so that $\tau = \tau_1 \tau_2$. Note that the walk $\tau_2$ can be written as a series of edges $\tau_2 = e_{\tau_2}^{(1)} e_{\tau_2}^{(2)} \cdots e_{\tau_2}^{(\ell (\tau_2) )}$. Then
    \[
    \prod_{i=1}^{\ell (\tau_2)} e_{\tau_2}^{(i)} = \frac{\mu \left( \tilde{\Sigma}_{i_{p^*} j_{p^*}}^{(\dep (v_{p^*}))} \right)}{\tilde{r}^{(\dep (v_{p^*}))}_{i_{p^*} j_{p^*}}} \mathbf{1}_{\sigma \in \tilde{\Sigma}^{(\dep (v_{p^*}))}_{i_{p^*} j_{p^*}}} f^{\nu (\tau)},
    \]
    with $\nu (\tau)$ defined in~\eqref{eq: nu_def_on_walks}.
\end{lemma}
The proof of Lemma~\ref{lem: prod_tau2_formulation} follows immediately from line~\ref{alg_line: E2_if_alpha} of Algorithm~\ref{alg: Whole_G_weight_assignment} along with~\cite[Lemma~7]{XC-BG:25}.

We next provide an expression for the entries of $\mathcal{C}(A,B)$ that will be used to analyze the $G_{\mathrm{cyc}}^+$ part of the graph, for the pair $(A,B)$ produced by Algorithm~\ref{alg: Whole_G_weight_assignment}.
\begin{lemma} \label{lem: C_ij_values}
    Suppose $j>n^*$, and $j > \max_{p=1}^{q} (d_p - \ell_p)$, and $\alpha_h \in G_p^+$. Let $C_{h,(j-1) m + i}$ be an entry of $\mathcal{C} (A,B)$. Then
    \[
    C_{h,(j-1) m + i} = \begin{cases}
        \int_{\Sigma} f^{\nu (\tau)} \mathbf{1}_{\sigma \in \tilde{\Sigma}^{(\dep (v_{p^*}))}_{i_{p^*} j_{p^*}}} d \mu & \text{if } \alpha_h \in V(i,j), \\ 
        0 & \text{otherwise},
    \end{cases}
    \]
    where $\tau$ is the unique walk of length $j$ from $\beta_{i}$ to $\alpha_h$ of the form~\eqref{eq: unique_walks_to_cyc+}. Moreover, note that $i = i_{p^*}$ due to the definition of $V(i,j)$.
\end{lemma}
\begin{proof}
    Let us start by the case where $\alpha_{h} \notin V(i,j)$. This means one of the two following cases may happen:
    
    \textbf{Case 1:} there does not exist any path of length $j$ from $\beta_i$ to $\alpha_h$. In this case, it is trivial that $C_{h,(j-1) m + i} = 0$.
    
    \textbf{Case 2:} there is a path of length $j$ from $\beta_i$ to $\alpha_h$, but it is not of the form~\eqref{eq: unique_walks_to_cyc+}. Note that, as indicated in~\eqref{eq: walks_to_cyc+}, since $G$ is reduced, all walks have to be of the form
    \begin{equation}
    \tau_{\beta_i \alpha_h} = \tau_{\beta_i v_{p^*}} \tau_{v_{p^*} \alpha_{p_0}} G_p^{k} \tau_{\alpha_{p_0} \alpha_h}. \label{eq: other_walks_to_be_shown_zero}
    \end{equation}
    However, since $\alpha_h \notin V(i,j)$, we know that
    \[
    \tau_{\beta_i v_{p^*}} \neq \tau_{\beta_{i_{p^*}} v_{p^*}}^*.
    \]
    As a result, $\tau_{\beta_i v_{p^*}}$ has to be part of some maximal path $p_{i j'}$ such that $(i,j') \neq (i_{p^*},j_{p^*})$. Now due to Algorithm~\ref{alg: core_weight_assignment}, we know that the $\Sigma_{i j'}$ and $\Sigma_{i_{p^*},j_{p^*}}$ are chosen to be disjoint. For simplicity, we let 
    \[
    \tau_{1} := \tau_{\beta_i v_{p^*}} \quad \text{and} \quad \tau_2 := \tau_{v_{p^*} \alpha_{p_0}} G_p^{k} \tau_{\alpha_{p_0} \alpha_h},
    \]
    so that $\tau_{\beta_i \alpha_h} = \tau_1 \tau_2.$
    Now recall that, by Algorithm~\ref{alg: core_weight_assignment}, the product of the edges on the path $\tau_{\beta_i v_{p^*}}$ will have an indicator term of $\mathbf{1}_{\tilde{\Sigma}_{i j'}^{(l)}}$ for some $l$, where $\tilde{\Sigma}_{i j'}^{(l)} \subseteq \Sigma_{i j'}$. Moreover, by Lemma~\ref{lem: prod_tau2_formulation}, the product of the edges on $\tau_2$ has an indicator term of $\mathbf{1}_{\sigma \in \tilde{\Sigma}^{(\dep (v_{p^*}))}_{i_{p^*} j_{p^*}}}$, where $\tilde{\Sigma}^{(\dep (v_{p^*}))}_{i_{p^*} j_{p^*}} \subseteq \Sigma_{i_{p^*} j_{p^*}}$. Now since $\Sigma_{i j'} \cap \Sigma_{i_{p^*},j_{p^*}} = \varnothing$, the product of the edges on $\tau_{\beta_i \alpha_h}$ will be zero, and hence, 
    \[
    C_{h,(j-1) m + i} = 0.
    \]
    
    All that is left to show is the case $\alpha_h \in V(i,j)$. Note that we can write $\tau$ as
    \[
    \tau = \tau_{\beta_{i_{p^*}} v_{p^*}}^* \tau_{v_{p^*} \alpha_{p_0}} G_p^{k} \tau_{\alpha_{p_0} \alpha_j}.
    \]
    Again, we denote
    \[
    \tau_{1} := \tau_{\beta_{i_{p^*}} v_{p^*}}^* \quad \text{and} \quad \tau_2 := \tau_{v_{p^*} \alpha_{p_0}} G_p^{k} \tau_{\alpha_{p_0} \alpha_h},
    \]
    so that $\tau = \tau_1 \tau_2$. Recall that we can write $\tau_1$ and $\tau_2$ as a series of edges $\tau_1 = e_{\tau_1}^{(1)} e_{\tau_1}^{(2)} \cdots e_{\tau_1}^{(\ell (\tau_1) )}$, and $\tau_2 = e_{\tau_2}^{(1)} e_{\tau_2}^{(2)} \cdots e_{\tau_2}^{(\ell (\tau_2) )}$. We now have the following sublemma:
    \begin{sublemma} \label{sublem: r_tilde_non_exp_form}
        For any given $i \in \{1,2,\dotsc,m\}$ and $j \in \{1,2,\dotsc,s_i\}$, it holds for all $k \in \{0,1,\dotsc,\delta_{ij}\}$ that
    \begin{equation*}
    \prod_{{l}=0}^k \tilde{e}_{ij}^{({l})} (\sigma) \; {\bf 1}_{\sigma \in \tilde{\Sigma}^{(k)}_{ij}} = \frac{\tilde{r}_{ij}^{(k)}}{\mu \left( \tilde{\Sigma}^{(k)}_{ij} \right)} {\bf 1}_{\sigma \in \tilde{\Sigma}^{(k)}_{ij}}.
    \end{equation*}
    \end{sublemma}
    \begin{proof}[Proof of Sublemma~\ref{sublem: r_tilde_non_exp_form}] From the weight assignment in Algorithm~\ref{alg: core_weight_assignment}, $\prod_{{l}=0}^k \tilde{e}_{ij}^{({l})} (\sigma)$ takes a constant value when $\sigma \in \tilde{\Sigma}^{(k)}_{ij}$, i.e.,
    \[
    \prod_{{l}=0}^k \tilde{e}_{ij}^{({l})} (\sigma) {\bf 1}_{\sigma \in \tilde{\Sigma}^{(k)}_{ij}} = c\,{\bf 1}_{\sigma \in \tilde{\Sigma}^{(k)}_{ij}},
    \]
    for some constant $c$. Now recall from Sublemma~\ref{sublem: r_value_induc} that
    \begin{align*}
    \tilde{r}_{ij}^{(k)} = \E \left[ \prod_{{l}=0}^k \tilde{e}_{ij}^{({l})} (\sigma) {\bf 1}_{\sigma \in \tilde{\Sigma}^{(k)}_{ij}} \right] = \E \left[ c\,{\bf 1}_{\sigma \in \tilde{\Sigma}^{(k)}_{ij}} \right] = c \ \mu \left( \tilde{\Sigma}^{(k)}_{ij} \right),
    \end{align*}
    and hence, we find the value of $c$ as
    \[
    c = \frac{\tilde{r}_{ij}^{(k)}}{\mu \left( \tilde{\Sigma}^{(k)}_{ij} \right)},
    \]
    finishing the proof of Sublemma~\ref{sublem: r_tilde_non_exp_form}. \end{proof}

    We can now write
    \begin{align*}
        C_{h,(j-1) m + i} &= \E \left[ \prod_{i=1}^{\ell (\tau)} e_{\tau}^{(i)}  \right] \\
        &= \E \left[ \left(\prod_{i=1}^{\ell (\tau_1)} e_{\tau_1}^{(i)} \right) \left(\prod_{i=1}^{\ell (\tau_2)} e_{\tau_2}^{(i)} \right) \right] \\
        &\overset{\mathrm{(i)}}{=} \E \left[ \left(\prod_{i=1}^{\ell (\tau_1)} e_{\tau_1}^{(i)} \right) \frac{\mu \left( \tilde{\Sigma}_{i_{p^*} j_{p^*}}^{(\dep (v_{p^*}))} \right)}{\tilde{r}^{(\dep (v_{p^*}))}_{i_{p^*} j_{p^*}}} \mathbf{1}_{\sigma \in \tilde{\Sigma}^{(\dep (v_{p^*}))}_{i_{p^*} j_{p^*}}} f^{\nu (\tau)} \right] \\
        &\overset{\mathrm{(ii)}}{=} \E \left[ \frac{\tilde{r}^{(\dep (v_{p^*}))}_{i_{p^*} j_{p^*}}}{\mu \left( \tilde{\Sigma}_{i_{p^*} j_{p^*}}^{(\dep (v_{p^*}))} \right)} \frac{\mu \left( \tilde{\Sigma}_{i_{p^*} j_{p^*}}^{(\dep (v_{p^*}))} \right)}{\tilde{r}^{(\dep (v_{p^*}))}_{i_{p^*} j_{p^*}}} \mathbf{1}_{\sigma \in \tilde{\Sigma}^{(\dep (v_{p^*}))}_{i_{p^*} j_{p^*}}} f^{\nu (\tau)} \right] \\
        &= \E \left[ f^{\nu (\tau)} \mathbf{1}_{\sigma \in \tilde{\Sigma}^{(\dep (v_{p^*}))}_{i_{p^*} j_{p^*}}} \right] \\
        &= \int_{\Sigma} f^{\nu (\tau)} \mathbf{1}_{\sigma \in \tilde{\Sigma}^{(\dep (v_{p^*}))}_{i_{p^*} j_{p^*}}} d \mu,
    \end{align*}
    where (i) follows from Lemma~\ref{lem: prod_tau2_formulation}, and (ii) from Sublemma~\ref{sublem: r_tilde_non_exp_form}. This finishes the proof.
\end{proof}



Let $C_j$ be the $j$-th column of $C$. We define the support of $C_j$ as $\supp (C_j) := \{ \alpha_h \in V_{\alpha} : C_{hj} \neq 0 \}$. By Lemma~\ref{lem: C_ij_values} and the strict positivity of the integral appearing there, we have that
\begin{equation}
\supp (C_{(j-1)m + i}) = V(i,j), \label{eq: support_connection_to_V_ij}
\end{equation}
for all $i \in \{1,\dotsc,m\}$ and $j \ge 1$. Let $j^*$ be such that
\[
j^* > n^*, \quad \text{and} \quad j^* > \max_{p=1}^{q} (d_p - \ell_p),
\]
so it satisfies the conditions in Lemma~\ref{lem: xc_bg_lemma5}. Now consider the columns $C_{(j^* + kL -1)m + i}$ for $k \in \N$. By Lemma~\ref{lem: xc_bg_lemma5} and~\eqref{eq: support_connection_to_V_ij}, we know that these columns share the same support. Let $V' = \{ \alpha_1', \dotsc, \alpha_{n_{j^*}}'\}$ be the common support, which is a subset of $V_{\text{cyc}}^+$.
Now let $C'$ denote the submatrix of $C$ by first taking the columns $C_{(j^* + kL -1)m + i}$ for all $k \in \N$, and then removing the zero rows, i.e.,
\begin{equation}
C' = \left[ C_{(j^* -1)m + i}, C_{(j^* + L  -1)m + i}, C_{(j^* + 2L -1)m + i} , \cdots \right]\Big|_{V'}. \label{eq: C'_Formulation}
\end{equation}
Let $C'_{hk}$ be the $(h,k)$-th entry of $C'$. We can express $C'_{hk}$ explicitly. Since $\alpha_{h}' \in V'$, there exists a unique walk of length $(j^* + (k-1)L)$, denoted by $\tau_{hk}$, which takes the form~\eqref{eq: unique_walks_to_cyc+}. Now note that $\tau_{hk}$ can
be obtained from $\tau_h:=\tau_{h1}$ by inserting the closed walk $G_p^{(k-1) L / \ell_p}$ for some $p = 1, \dotsc, q$, i.e.,
\[
\tau_{h k} := \tau_{\beta_{i_{p^*}} v_{p^*}}^* \tau_{v_{p^*} \alpha_{p_0}} G_p^{(k-1)L/\ell_p} \tau_{\alpha_{p_0} \alpha_h}.
\]
Therefore, by~\eqref{eq: nu_def_on_walks}, 
\[
\nu (\tau_{hk}) = \nu (\tau_{h}) + \frac{\ell (\tau_{hk}) - \ell (\tau_h)}{L} = \nu (\tau_h) + k-1.
\]
Hence, by Lemma~\ref{lem: C_ij_values}, we have that
\begin{equation}
    C'_{hk} = \int_{K} f^{\nu (\tau_h) + k-1} \mathbf{1}_{\sigma \in \hat{\Sigma}_h} d \mu, \label{eq: C'_hk_values}
\end{equation}
where we use $\hat{\Sigma}_h$ for ease of notation, i.e., for $\alpha_h' \in G_p$, $\hat{\Sigma}_h = \tilde{\Sigma}^{(\dep (v_{p^*}))}_{i_{p^*} j_{p^*}}$ as discussed in Lemma~\ref{lem: C_ij_values}. In particular, every $\hat{\Sigma}_h$ is a positive-measure subset of $K$, so the integral over $\Sigma$ in Lemma~\ref{lem: C_ij_values} is equal to the integral over $K$ in~\eqref{eq: C'_hk_values}. Moreover, note that $\nu (\tau_1), \dotsc, \nu (\tau_{n_{j^*}})$ are pairwise distinct; this can be shown similarly to~\cite[Lemma~9]{XC-BG:25}.
Let us now provide the following remark on the choice of these $\hat{\Sigma}_h$s that is done in Algorithm~\ref{alg: core_weight_assignment} and is later also used in Algorithm~\ref{alg: Whole_G_weight_assignment}:

\begin{remark}[Radial realization of the sets $\hat{\Sigma}_h$ and laminarity]\label{rem:radial_realization_hatSigma}
In the run of Algorithm~\ref{alg: core_weight_assignment} used by Algorithm~\ref{alg: Whole_G_weight_assignment}, the active set is $D=K$.
By Lemma~\ref{lem:radial_equal_measure_splitting}, the equal-measure partitions in this run can be chosen using intervals of the radial map $f$. Moreover, if an interval $I\subseteq[0,1]$ is split at a threshold $c$, the two pieces are
\[
I\cap[0,c]
\quad\text{and}\quad
I\cap(c,1].
\]
Thus the pieces are disjoint and their union is exactly $I$. For instance, splitting $(0.3,1]$ at $0.5$ gives $(0.3,0.5]$ and $(0.5,1]$. Since every level set $K\cap f^{-1}(\{c\})$ has measure zero, this endpoint convention does not impact the required measures.

\noindent Hence every set generated in this run can be written exactly as
\[
K\cap f^{-1}(I),
\]
for some interval $I\subseteq[0,1]$. In particular, for each set $\hat\Sigma_h$ appearing in~\eqref{eq: C'_hk_values}, there exists an interval $I_h\subseteq[0,1]$ such that
\[
\hat\Sigma_h=K\cap f^{-1}(I_h).
\]
Consequently, $\mathbf 1_{\sigma \in \hat\Sigma_h}$ is measurable with respect to
\[
\bm{\sigma}(f)|_K
:=
\{K\cap f^{-1}(Z):Z\in\mathcal B([0,1])\}.
\]
Moreover, the family $\{\hat\Sigma_h\}$ is laminar. Sets generated on different branches are disjoint, while sets generated on the same branch are nested or equal. Therefore, the corresponding intervals $\{I_h\}$ may also be chosen to be disjoint, nested, or equal. \oprocend
\end{remark}

With Remark~\ref{rem:radial_realization_hatSigma} in place, we are now in the position to establish the following result:
\begin{lemma} \label{lem: submatrix_full_rank}
    The submatrix $C'$ given in~\eqref{eq: C'_Formulation} is of full row rank.
\end{lemma}
\begin{proof}
To show that $C'$ has full row rank, it suffices to prove that $y^\top C'=0$ for
\[
y=[y_1,\dotsc,y_{n_{j^*}}]^\top\in\mathbb R^{n_{j^*}}
\]
implies $y=0$. Let $y^\top C'=0$. By~\eqref{eq: C'_hk_values}, for every $k\in\{1,2,\dotsc\}$,
\begin{equation}
0=\sum_{h=1}^{n_{j^*}}y_h C'_{hk}
=\sum_{h=1}^{n_{j^*}}y_h\int_K
f^{\nu(\tau_h)+k-1}\mathbf 1_{\sigma\in\hat\Sigma_h}\,d\mu.
\label{eq:proof_row_dependence_start}
\end{equation}
After reindexing, this is equivalent to
\begin{equation}
\int_K
\left(\sum_{h=1}^{n_{j^*}}y_h f^{\nu(\tau_h)}
\mathbf 1_{\sigma\in\hat\Sigma_h}\right)f^k\,d\mu=0,
\label{eq:moment_orthogonality}
\end{equation}
for all $k\in\N$. Define on $K$
\[
H(\sigma):=\sum_{h=1}^{n_{j^*}}y_h f(\sigma)^{\nu(\tau_h)}
\mathbf 1_{\sigma\in\hat\Sigma_h}.
\]
By linearity,~\eqref{eq:moment_orthogonality} gives
\[
\int_K H(\sigma)p(f(\sigma))\,d\mu=0
\]
for every polynomial $p$. Now let
\[
\rho:=f_{\#}(\mu|_K)
\]
be the pushforward of the restricted finite measure $\mu|_K$.
By Remark~\ref{rem:radial_realization_hatSigma}, for each $h$ there is an interval $I_h\subseteq[0,1]$ such that
\[
\hat\Sigma_h=K\cap f^{-1}(I_h).
\]
Define
\[
\widetilde H(t):=
\sum_{h=1}^{n_{j^*}}y_h t^{\nu(\tau_h)}\mathbf 1_{t\in I_h}.
\]
Then, for every $\sigma\in K$,
\begin{align*}
H(\sigma) =\sum_{h=1}^{n_{j^*}}y_h f(\sigma)^{\nu(\tau_h)} \mathbf 1_{f(\sigma)\in I_h} = \widetilde H(f(\sigma)).
\end{align*}
Thus, $H$ is $\bm{\sigma}(f)|_K$-measurable; this is an explicit instance of the Doob--Dynkin lemma~\cite{DW:91}.
Therefore,
\[
\int_{[0,1]}\widetilde H(t)p(t)\,d\rho(t)=0
\]
for every polynomial $p$. 
Since the exponents $\nu(\tau_h)$ are positive and the sum is finite, $\widetilde H\in L^2(\rho)$. Since polynomials are dense in $L^2(\rho)$ and $\widetilde H$ is orthogonal to every polynomial, we obtain $\widetilde H=0$ in $L^2(\rho)$, i.e.,
\[
\int_{[0,1]}|\widetilde H(t)|^2\,d\rho(t)=0.
\]
Consequently,
\begin{equation}
\widetilde H(t)=0
\quad\rho\text{-a.e. on }[0,1].
\label{eq:H_zero_rho_ae}
\end{equation}
We next record the support property of $\rho$. If $I\subseteq(0,1)$ is a nonempty open interval, then
\[
f^{-1}(I)\cap K
=
\phi^{-1}\big(\{z\in\mathbb B:\|z\|\in I\}\big)
\]
is a nonempty open subset of $\mathcal U\subseteq\operatorname{supp}\mu$. Hence
\[
\rho(I)=\mu(f^{-1}(I)\cap K)>0.
\]
It follows that
\begin{equation}
(0,1)\subseteq\operatorname{supp}\rho.
\label{eq:rho_full_radial_support}
\end{equation}
Let $\mathcal E$ be the finite set consisting of $0$, $1$, and all endpoints of the intervals $I_h$. Write its elements in increasing order as
\[
0=t_0<t_1<\cdots<t_N=1.
\]
Let $\mathfrak R$ denote the collection of nonempty open interval cells
\[
(t_{\ell-1},t_\ell),\qquad \ell=1,\dotsc,N.
\]
For each $\cR\in\mathfrak R$, the following set of indices
\[
J(\cR):=\{h:\cR\subseteq I_h\}
\]
is fixed, and~\eqref{eq:H_zero_rho_ae} reduces to
\begin{equation}
Q_{\cR}(t):=\sum_{h\in J(\cR)}y_h t^{\nu(\tau_h)}=0
\quad\rho\text{-a.e. on }\cR.
\label{eq:qR_zero_ae}
\end{equation}
The function $Q_{\cR}$ is continuous on $\cR$. By~\eqref{eq:rho_full_radial_support}, every nonempty open subinterval of $\cR$ has positive $\rho$-measure. Therefore~\eqref{eq:qR_zero_ae} implies that $Q_{\cR}$ vanishes at every point of $\cR$; otherwise, continuity would give a nonempty open subinterval on which $Q_{\cR}$ is bounded away from zero, contradicting~\eqref{eq:qR_zero_ae}.

The exponents $\nu(\tau_h)$ are pairwise distinct; thus, the functions $\left\{t^{\nu(\tau_h)}\right\}_{h\in J(\cR)}$ are linearly independent on every nonempty interval contained in $(0,1)$.
Hence
\[
y_h=0,
\]
for all $h\in J(\cR)$.
Finally, every $\hat\Sigma_h$ has positive measure by construction, and
\[
\rho(I_h)=\mu(\hat\Sigma_h)>0.
\]
Now by Lemma~\ref{lem:radial_equal_measure_splitting}, $\rho$ is nonatomic, so $I_h$ cannot be a singleton. Therefore, $I_h$ contains at least one of the nondegenerate cells $\cR \in \mathfrak R$, and hence $h\in J(\cR)$. Applying the preceding conclusion to that cell gives $y_h=0$. Since this holds for every $h \in \{1,\dotsc,n_{j^*}\}$, we get $y=0$. As a result, the rows of $C'$ are linearly independent, and $C'$ has full row rank.
\end{proof}

With Lemma~\ref{lem: submatrix_full_rank} in place, we can now finish the sufficiency argument.

\begin{proposition}\label{prop:sufficiency_full_rank}
Let $(A,B)$ be the pair constructed by Algorithm~\ref{alg: Whole_G_weight_assignment}. Then the averaged controllability matrix $\mathcal{C}(A,B)$ has full row rank.
\end{proposition}

\begin{proof}
Recall that $n^* = |V_\alpha^*|$ and let $n^+ := |V_{\mathrm{cyc}}^+| = n - n^*$. By construction, the core assignment realizes a full-row-rank averaged controllability matrix on $V_\alpha^*$. Hence, we may fix $n^*$ columns of $\mathcal{C}(A,B)$ whose restriction to the rows indexed by $V_\alpha^*$ is linearly independent (with this collection empty when $n^*=0$).

It remains to find $n^+$ columns that are linearly independent on the rows indexed by $V_{\mathrm{cyc}}^+$. For $i\in\{1,\dots,m\}$ and $j\ge 1$, recall from~\eqref{eq: support_connection_to_V_ij} that
\[
\supp\!\big(C_{(j-1)m+i}\big)=V(i,j),
\]
as before.
Let $j^*$ satisfy the conditions of Lemma~\ref{lem: xc_bg_lemma5}. For each $i\in\{1,\dots,m\}$ and each $\ell\in\{0,1,\dots,\ell_{\max}^{(i)} -1\}$, define
\[
V'_{\ell}(i) := V(i,j^*+\ell)\setminus \bigcup_{k=0}^{\ell-1}V(i,j^*+k),
\]
By definition, the sets $\{V'_{\ell}(i)\}_{\ell=0}^{\ell_{\max}^{(i)}-1}$ are disjoint for each fixed $i$. Moreover, by~\eqref{eq: disjoint_feeding_to_cyc_plus}, the sets $V(i,\cdot)$ and $V(i',\cdot)$ are disjoint whenever $i\neq i'$, and therefore $V'_{\ell}(i)\cap V'_{\ell'}(i')=\emptyset$ for all $\ell,\ell'$ and all $i\neq i'$. In particular, the family $\{V'_{\ell}(i)\}_{i,\ell}$ is pairwise disjoint (allowing some sets to be empty).

By Lemma~\ref{lem: xc_bg_lemma5}, the union of the sets $\{V(i,j^*+\ell)\}_{i,\ell}$ covers $V_{\mathrm{cyc}}^+$. Consequently,
\[
V_{\mathrm{cyc}}^+ = \bigcup_{i=1}^m \bigcup_{\ell=0}^{\ell_{\max}^{(i)} -1} V'_{\ell}(i),
\]
and hence
\[
n^+ = |V_{\mathrm{cyc}}^+|
= \sum_{i=1}^m \sum_{\ell=0}^{\ell_{\max}^{(i)}-1} |V'_{\ell}(i)|.
\]

Fix $i$ and $\ell$. Lemma~\ref{lem: xc_bg_lemma5} implies that the supports of the columns
\[
\Big\{\,C_{(j^*+\ell+kL-1)m+i}\,:\,k\in\N\,\Big\}
\]
are all equal to $V(i,j^*+\ell)$, hence they all contain $V'_{\ell}(i)$. Let $C'_{\ell}(i)$ be the submatrix obtained by taking these columns and restricting to the rows indexed by $V'_{\ell}(i)$ (equivalently, removing all-zero rows). By Lemma~\ref{lem: submatrix_full_rank}, the matrix $C'_{\ell}(i)$ has full row rank, namely
\[
\rank\big(C'_{\ell}(i)\big)=|V'_{\ell}(i)|.
\]
So for each pair $(i,\ell)$, we can select $|V'_{\ell}(i)|$ columns from $\{C_{(j^*+\ell+kL-1)m+i}\}_{k\in\N}$ whose restriction to $V'_{\ell}(i)$ is linearly independent.

Since the sets $V'_{\ell}(i)$ are pairwise disjoint across all $(i,\ell)$, the union of these selected columns is linearly independent on the rows indexed by $\bigcup_{i,\ell}V'_{\ell}(i)$, and thus contains a subcollection of $n^+$ columns that is linearly independent on $V_{\mathrm{cyc}}^+$. Finally, because $j^*>n^*$, these columns supported outside the core have support disjoint from the core rows, and hence together with the previously fixed $n^*$ core columns they form a set of $n^*+n^+=n$ linearly independent columns of $\mathcal{C}(A,B)$. This proves that $\mathcal{C}(A,B)$ has full row rank, concluding the proof.
\end{proof}

\section{Conclusion}\label{sec:conclusion}

We have established a necessary and sufficient graph-theoretic condition for structural averaged controllability of multi-input linear ensemble systems. The condition has two parts: every state node must be reachable from an input node, and the row-column bipartite graph associated with the core must admit a row-saturating matching. This completes the multi-input extension of the single-input characterization in~\cite{XC-BG:25}.

The proof is constructive. On the core, we have shown that any desired averaged controllability matrix, as long as it fits the sparsity pattern, can be assigned using suitable bounded measurable edge weights. We then extended this assignment to the rest of the graph, using a scalar radial function on the parameter space, while preserving the rank needed for averaged controllability.

Several potential future directions remain. First, our construction guarantees a full-row-rank averaged controllability matrix, but it is not necessarily minimal; it may be possible to achieve the required rank using fewer block columns. Another direction is to ask what happens beyond averaged controllability, for example by studying structural conditions for controlling other statistics of the ensemble, such as higher moments.

\section*{Acknowledgments}
This research was supported by the National Science Foundation under the collaborative research grant 2303221.


\end{document}